%% file: sigmod.tex
\documentclass[acmsmall]{acmart}
\AtBeginDocument{%
  }

\acmJournal{JACM}
\acmVolume{37}
\acmNumber{4}
\acmArticle{111}
\acmMonth{8}
\input{def/yf-def}

\usepackage{amsthm}  
\usepackage{algorithm} 
\usepackage{algpseudocode} 
\usepackage{xcolor} 
\usepackage{tabularx}
\usepackage{subcaption}
\usepackage{amsmath}  
\usepackage{booktabs} 
\usepackage{makecell}
\usepackage{graphicx}
\usepackage{tikz}
\usetikzlibrary{patterns}
\usetikzlibrary{arrows.meta, calc, decorations.pathreplacing}
\usetikzlibrary{patterns.meta}
\usepackage{stfloats} 
\usepackage{tcolorbox}
\usepackage{pifont}
\usepackage{utfsym}
\usepackage{caption}  
\tcbuselibrary{skins} 
\newcolumntype{M}{>{\centering\arraybackslash}m{2cm}}
\newcolumntype{L}{>{\raggedright\arraybackslash}m{6.0cm}}

\newcommand{\REV}[1]{{#1}}

\def\vgap{\vspace{0mm}}
\def\extraspacing{\vspace{0mm} \noindent}
\def\figcapup{\vspace{-2mm}}
\def\figcapdown{\vspace{-2mm}}

\def\dis{\mit{\delta}}
\def\X{\mathcal{X}}
\def\V{\mathcal{V}}
\def\E{\mathcal{E}}
\def\diam{\mit{diam}}
\def\dim{d}

\begin{document}
\setcopyright{cc}
\setcctype{by}
\acmJournal{PACMMOD}
\acmYear{2026}
\acmVolume{4}
\acmNumber{1 (SIGMOD)}
\acmArticle{36}
\acmMonth{2}
\acmPrice{}
\acmDOI{10.1145/3786650}
\title{Fast-Convergent Proximity Graphs for Approximate Nearest Neighbor Search}

\author{Binhong Li}
\affiliation{%
  \institution{DSA Thrust, The Hong Kong University of Science and Technology (Guangzhou)}
  \country{China}
}
\email{bli120@connect.hkust-gz.edu.cn}

\author{Xiao Yan}
\affiliation{%
  \institution{Institute for Math \& AI, Wuhan University}
  \country{China}
}
\email{yanxiaosunny@whu.edu.cn}

\author{Shangqi Lu}
\authornote{Dr. Shangqi Lu is the corresponding author.}
\affiliation{%
  \institution{DSA Thrust, The Hong Kong University of Science and Technology (Guangzhou)}
  \country{China}
}
\email{shangqilu@hkust-gz.edu.cn}


\begin{abstract}
Approximate nearest neighbor (ANN) search in high-dimensional metric spaces is a fundamental problem with many applications. Over the past decade, proximity graph (PG)-based indexes have demonstrated superior empirical performance over alternatives. However, these methods often lack theoretical guarantees regarding the quality of query results, especially in the worst-case scenarios. In this paper, we introduce the \emph{$\alpha$-convergent graph} ($\alpha$-CG), a new PG structure that employs a new carefully designed \emph{edge pruning rule}. If the distance between the query point $q$ and its exact nearest neighbor $v^*$ is at most $\tau$ for some constant $\tau > 0$, our $\alpha$-CG finds the exact nearest neighbor in poly-logarithmic time, assuming bounded intrinsic dimensionality for the dataset; otherwise, it can find an ANN in the same time. To enhance scalability, we develop the \emph{$\alpha$-convergent neighborhood graph} ($\alpha$-CNG), a practical variant that applies the pruning rule locally within each point’s neighbors. We also introduce optimizations to reduce the index construction time. Experimental results show that our $\alpha$-CNG outperforms existing PGs on real-world datasets. For most datasets, $\alpha$-CNG can reduce the number of distance computations and search steps by over 15\% and 45\%, respectively, when compared with the best-performing baseline.
\end{abstract}



\received{July 2025}
\received[revised]{October 2025}
\received[accepted]{November 2025}
\maketitle
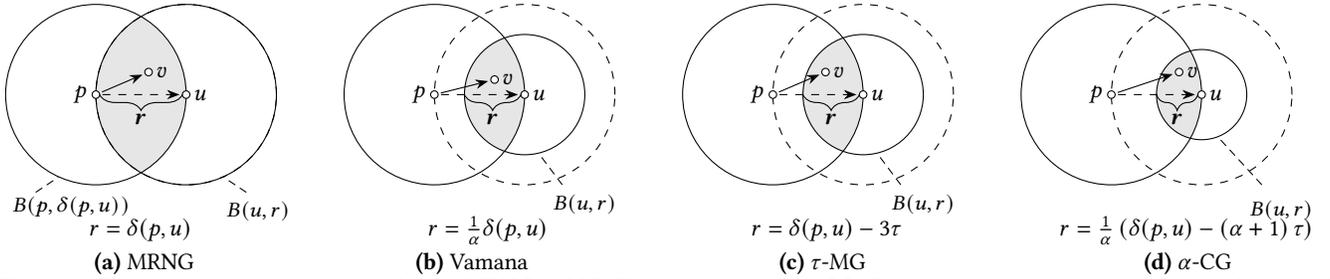
\begin{figure*}
	\centering
	
	\begin{tikzpicture}[scale=0.8, every node/.style={font=\small}]
		
        \begin{scope}[shift={(0,0)}]
			\coordinate (u1) at (0,0);
			\coordinate (v1) at (1.2,0);
			\coordinate (up) at ($(v1)+(-0.5,0.3)$);
			\begin{scope}
				  \clip (u1) circle(1.2); 
				  \fill[gray!20] (v1) circle(1.2);
				\end{scope}
			
			\draw (u1) circle(1.2);
			\draw [dashed] (v1) circle(1.2);
			\draw (v1) circle(1.2);
			\draw[dashed, ->, >=Stealth,shorten <=2.5pt, shorten >=2.5pt] (u1) -- (v1);
                \draw[fill=white] (u1) circle(0.05) node[left] {$p$};
                \draw[fill=white] (v1) circle(0.05) node[right] {$u$};
			\draw[fill=white] (up) circle(0.05) node[right] {$v$};
			\draw [->, >=Stealth,shorten <=2.5pt ,shorten >=2.5pt] (u1) -- (up);
                \draw[decorate, decoration={brace, mirror, amplitude=6pt}]
                ($(u1)+(0.05,0)$) -- ($(v1)+(-0.05,0)$)
                node[midway, below=4pt] {\normalsize $\boldsymbol{r}$};
			\draw[densely dashed, line width=0.4pt, >=Stealth]
			($(u1)-(0.8,1.3)$) -- ($(u1)-(0.5,1.1)$)
			node[below,xshift=5pt, yshift=-4pt] {\small $B(p,\dis(p,u))$};
			\draw[densely dashed, line width=0.4pt, >=Stealth]
			($(v1)+(0.8,-1.3)$) -- ($(v1)+(+0.5,-1.1)$)
			node[below right,xshift=-2pt, yshift=-6pt] {\small $B(u,r)$};
			\node[anchor=west] at ($(u1)+(-0.4,-2.3)$) {\normalsize $r = \dis(p,u)$};
		\end{scope}

		\begin{scope}[shift={(4.5,0)}]
			\coordinate (u1) at (0,0);
			\coordinate (v1) at (1.2,0);
			\coordinate (p3t1) at (0.4,0);
			\coordinate (up) at ($(v1)+(-0.4,0.2)$);
			\begin{scope}
				  \clip (u1) circle(1.2);
				  \fill[gray!20] (v1) circle(0.8);
				\end{scope}
			\draw (u1) circle(1.2);
			\draw [dashed] (v1) circle(1.2);
			\draw (v1) circle(0.8);
			\draw[dashed, ->, >=Stealth,shorten <=2.5pt, shorten >=2.5pt] (u1) -- (v1);
                \draw[fill=white] (u1) circle(0.05) node[left] {$p$};
                \draw[fill=white] (v1) circle(0.05) node[right] {$u$};
			\draw[fill=white] (up) circle(0.05) node[right] {$v$};
			\draw [->, >=Stealth,shorten <=2.5pt ,shorten >=2.5pt] (u1) -- (up);
			\draw[decorate, decoration={brace, mirror, amplitude=6pt}]
			(p3t1) -- ($(v1)+(-0.05,0)$)
			node[midway, below=4pt] {\normalsize $\boldsymbol{r
				}$};
			\draw[densely dashed, line width=0.4pt, >=Stealth]
			($(v1)+(0.5,-1.2)$) -- ($(v1)+(+0.2,-0.8)$)
			node[below right,xshift=2pt, yshift=-12pt] {\small $B(u,r)$};
			\node[anchor=west] at ($(u1)+(-0.6,-2.3)$) {\normalsize $r =\frac{1}{\alpha}\dis(p,u)$};
		\end{scope}

		\begin{scope}[shift={(9,0)}]
			\coordinate (u1) at (0,0);
			\coordinate (v1) at (1.2,0);
			\coordinate (p3t1) at ($(u1)!0.333!(v1)$);
			\coordinate (up) at ($(v1)+(-0.5,0.3)$);
			\begin{scope}
				  \clip (u1) circle(1.2); 
				  \fill[gray!20] (v1) circle(0.8);
				\end{scope}
			
			\draw (u1) circle(1.2);
			\draw [dashed] (v1) circle(1.2);
			\draw (v1) circle(0.8);
			\draw[dashed, ->, >=Stealth,shorten <=2.5pt, shorten >=2.5pt] (u1) -- (v1);
                \draw[fill=white] (u1) circle(0.05) node[left] {$p$};
                \draw[fill=white] (v1) circle(0.05) node[right] {$u$};
			\draw[fill=white] (up) circle(0.05) node[right] {$v$};
			\draw [->, >=Stealth,shorten <=2.5pt ,shorten >=2.5pt] (u1) -- (up);
                \draw[decorate, decoration={brace, mirror, amplitude=6pt}]
			(p3t1) -- ($(v1)+(-0.05,0)$)
			node[midway, below=4pt] {\normalsize $\boldsymbol{r
				}$};
			\draw[densely dashed, line width=0.4pt, >=Stealth]
			($(v1)+(0.5,-1.2)$) -- ($(v1)+(+0.2,-0.8)$)
			node[below right,xshift=2pt, yshift=-12pt] {\small $B(u,r)$};
			\node[anchor=west] at ($(u1)+(-0.8,-2.3)$) {\normalsize $r = \dis(p,u) - 3\tau$};
		\end{scope}
		
		\begin{scope}[shift={(13.5,0)}]
			\coordinate (u1) at (0,0);
			\coordinate (v1) at (1.2,0);
			\coordinate (p3t1) at (0.3,0);
                \coordinate (p3t2) at (0.6,0);
			\coordinate (up) at ($(v1)+(-0.3,0.3)$);
			\begin{scope}
				  \clip (u1) circle(1.2);
				  \fill[gray!20] (v1) circle(0.6);
				\end{scope}
			
			
			\draw (u1) circle(1.2);
			\draw [dashed] (v1) circle(1.2);
			\draw (v1) circle(0.6);
			\draw[dashed, ->, >=Stealth,shorten <=2.5pt, shorten >=2.5pt] (u1) -- (v1);
                \draw[fill=white] (u1) circle(0.05) node[left] {$p$};
                \draw[fill=white] (v1) circle(0.05) node[right] {$u$};
			\draw[fill=white] (up) circle(0.05) node[right] {$v$};
			\draw [->, >=Stealth,shorten <=2.5pt ,shorten >=2.5pt] (u1) -- (up);
                \draw[decorate, decoration={brace, mirror, amplitude=6pt}]
			(p3t2) -- ($(v1)+(-0.05,0)$)
			node[midway, below=4pt] {\normalsize $\boldsymbol{r
				}$};
			\draw[densely dashed, line width=0.4pt, >=Stealth]
			($(v1)+(0.8,-1.2)$) -- ($(v1)+(+0.25,-0.6)$)
			node[below right, xshift=2pt, yshift=-16pt] {\small $B(u,r)$};
			\node[anchor=west] at ($(u1)+(-1.8,-2.3)$) {\normalsize $r=\frac{1}{\alpha}\left(\dis(p,u) - \left(\alpha+1\right)\tau\right)$};
		\end{scope}
	\end{tikzpicture}
	
	\noindent
        \hspace{0.1cm}
        \makebox[2.5cm][c]{\normalsize \textbf{(a)} MRNG}
	\hspace{0.8cm}
	\makebox[2.5cm][c]{\normalsize \textbf{(b)} Vamana}
	\hspace{0.8cm}
	\makebox[2.5cm][c]{\normalsize \textbf{(c)} $\tau$-MG}
	\hspace{0.8cm}
	\makebox[2.5cm][c]{\normalsize \textbf{(d)} $\alpha$-CG}
	\caption{The edge pruning rules of existing MRNG, $\tau$-MG, Vamana, and $\alpha$-CG. Given a data point $p$, a candidate $u$ is pruned if $p$ already has an out-neighbor $v$ falling within the intersection of $B(p, \dis(p, u))$ and $B(u, r)$.}
	\label{fig:lune-final}
\end{figure*}

\section{Introduction}
Approximate nearest neighbor search is a fundamental component in various applications such as image retrieval~\cite{kg09,pci+07}, recommendation systems~\cite{ddgr07,bnmn16}, and data mining~\cite{hff+17,i16}. Let $P$ be a set of $n$ data points from a metric space $(\X, \dis)$ where $\X$ is a set of points and $\dis:\X \times \X \rightarrow \real_{\ge 0}$ is a distance function. We want to construct a data structure on $P$ that supports the following {\em nearest neighbor} (NN) query: given a query point $q\in \X$, return the point $v^* \in P$ closest to $q$. Due to the ``curse of dimensionality", all known space-efficient solutions for NN search have search times that grow exponentially with the dimensionality of the metric space. To enhance query efficiency, researchers often resort to $c$-approximate nearest neighbor (ANN) queries for some constant $c > 1$, which balance query accuracy and search time by returning a point whose distance to $q$ is within a $c$-factor of the exact NN.


In recent years, numerous studies \cite{my20,fxwc19,pcc+23, sds+19,wxyw21,aep25} have demonstrated that proximity graph-based solutions exhibit superior empirical performance on real-world data over other indexes (e.g., trees~\cite{wwj+13,laja08} and inverted index~\cite{bl14,jds10}). A {\em proximity graph} (PG) is a simple directed graph, where each vertex represents a data point $p \in P$ and edges are connected based on particular geometric properties (usually close neighbors in distance). Given a query point $q$, a simple {\em greedy routing} is performed to find an ANN of $q$: starting from a fixed or random entry vertex, at each step, it explores the nearest out-neighbor of the currently visited vertex to $q$. This process traverses a search path until no closer nodes can be identified, returning the closest visited point as the answer\footnote{We can generalize the algorithm to find $k$ ANNs by maintaining a sorted list, known as the {\em beam search} algorithm (see Section~\ref{sec:pre:pg}).}. The running time of the algorithm is bounded by the total out-degree of the nodes on the search path. Therefore, to ensure efficient query performance, the goal is to construct a sparse graph (with a low maximum out-degree) that allows the search algorithm to converge to the ANNs of $q$ quickly.
\begin{table*}[!t]
\centering
\setlength{\tabcolsep}{3pt}  
\renewcommand{\arraystretch}{1.1} 
\figcapdown
\caption{Comparison of PG-based methods}
\figcapdown\figcapdown
\label{tab:method_comparison}

\resizebox{\textwidth}{!}{
\begin{tabular}{@{}l c c c c c c c@{}}
\toprule
\textbf{ } 
& \textbf{ }
& \multicolumn{3}{c}{\textbf{Accuracy Guarantees}} 
& \multicolumn{3}{c}{\textbf{Theoretical Properties}} \\
\cmidrule(lr){3-5} \cmidrule(lr){6-8}
\textbf{Method}
&\textbf{Pruning radius}
& $\delta(q, v^*) = 0$ 
& $\delta(q, v^*) < \tau$ 
& $\delta(q, v^*) > \tau$ 
& Routing Length 
& Avg. Time 
& Worst-case Time \\
\midrule

MRNG~\cite{fxwc19}
& $d(p,u)$
& \usym{1F5F8} & \usym{2715} & \usym{2715} 
& Linear 
& $O(n^{\frac{2}{m}}\ln n)$ 
& $O(n)$ \\

$\tau$-MG~\cite{pcc+23} 
& $d(p,u) - 3\tau$
& \usym{1F5F8} & \usym{1F5F8} & \usym{2715}
& Linear 
& $O(n^{\frac{1}{m}} (\ln n)^2)$ 
& $O(n)$ \\

Vamana~\cite{sds+19} 
& $\frac{1}{\alpha}d(p,u)$
& \usym{1F5F8} & \usym{2611} & \usym{2611}
& $O(\log_\alpha \fr{\Delta}{(\alpha-1)\eps})$ 
& -- 
& $O\!\left(\alpha^\dim \log \Delta \log \fr{\Delta}{(\alpha-1)\eps}\right)$ \\

\textbf{$\alpha$-CG (Ours)} 
& $\frac{1}{\alpha}d(p,u)-(\alpha+1)\tau$
& \usym{1F5F8} & \usym{1F5F8} & \usym{2611}
& $O(\log_\alpha \Delta)^\dag$ 
& -- 
& $O\!\left((\alpha \tau)^\dim \log \Delta \log_\alpha \Delta\right)^\dag$ \\

\bottomrule
\end{tabular}
}

\begin{minipage}{\textwidth}
\footnotesize
\raggedright
\usym{1F5F8}: exact guarantee; \hspace{1em}
\usym{2611}: $\left(\frac{\alpha+1}{\alpha-1} + \epsilon\right)$-ANN; \hspace{1em}
\usym{2715}: no guarantee; \hspace{1em}
$m$: dimension of Euclidean space; \hspace{1em}
$d$: doubling dimension; \hspace{1em}
$\Delta$: aspect ratio; \hspace{1em}
$^\dag$ when $\dis(q, v^*) \le \tau$.
\end{minipage}
\end{table*}
\subsection{Existing PGs and Their Limitations}

Most methods for constructing PGs follow a common framework. For each data point $p$, we first obtain a candidate set $\V$ of points from $P$ that are close to $p$, through greedy search on a base graph (e.g., an approximate $K$-NN graph \cite{fxwc19, pcc+23}) or an intermediate graph over a subset of $P$ \cite{my20, sds+19}. A small subset of $\V$ is then selected as the out-neighbors of $p$ to maintain graph connectivity. Crucially, the resulting graph should have the ``shortcutable'' property: for any query point $q$, if data point $p$ is not an ANN of $q$, $p$ should have an out-neighbor that is closer to $q$ than $p$. To achieve this, existing PG solutions introduce {\em edge pruning rules} to eliminate unnecessary candidates in $\V$ in the following procedure. For each point $u\in \V$ processed in ascending order of the distances to $p$:
\myitems{
	\item $u$ is pruned if there already exists an out-neighbor $v$ of $p$ falling within the intersection of the two balls\footnote{Given any point $p\in \X$ and $r > 0$, the {\em ball} $B(p, r)$ with radius $r$ represents the set containing all the points in $\X$ whose distance to $p$ is at most $r$.} $B(p, \dis(p,u))$ and $B(u, r)$, where $r > 0$ is the \REV{\em pruning radius} depending on $\dis(p, u)$;
    \item otherwise, $u$ is added as a new out-neighbor of $p$.
}

\REV{Different choices of the pruning radius $r$
lead to distinct properties and search guarantees of PGs.} In early PG methods, such as HNSW \cite{my20} and NSG \cite{fxwc19}, $r$ is set directly to $\dis(p, u)$, \REV{as illustrated in
Fig.~\ref{fig:lune-final}a.}
When $\V = P \setminus\set{p}$ for every $p\in P$ (the candidate set is the entire dataset), the constructed PG is the {monotonic relative neighborhood graph} (MRNG) and has the following property \cite{fxwc19, am93}: if query $q$ is a data point in $P$, a greedy routing will always terminate at $q$. The query time is $O(n^{2/m} \ln n)$ when the points in $P$ are from a uniform distribution in the $m$-dimensional Euclidean space \cite{fxwc19}. However, the search path length can be as large as $O(n)$ in the worst case, and when $q\notin P$, MRNG does not guarantee finding even an ANN.

\vgap


\REV{The Vamana graph in the DiskANN method \cite{sds+19} employs a more relaxed edge pruning rule by setting $r = \dis(p, u)/\alpha$ (see Figure~\ref{fig:lune-final}b) where $\alpha > 1$ is a parameter.}
When $\V = P\setminus \set{p}$ for each $p\in P$, we refer to the constructed PG as the {\em slow preprocessing version} of Vamana. It can be verified that the distance of every visited node on the search path to the query point decreases by an $\alpha$-multiplicative factor when $q\in P$, but this does not hold when $\dis(q, v^*) > 0$. Recently, Indyk and Xu \cite{ix23} studied the worst-case performance of popular PGs (such as HNSW, NSG, and Vamana) when $\dis(q, v^*)$ can be arbitrarily large. They proved that only (the slow preprocessing version of) Vamana provides approximation guarantees: a greedy routing finds an $(\fr{\alpha+1}{\alpha-1}+\eps)$-ANN of $q$ in $O(\alpha^\dim \cdot \log \Delta \cdot \log_\alpha \fr{\Delta}{(\alpha-1)\eps})$ time. Here, $\Delta$ is the {\em aspect ratio}\footnote{The ratio between the maximum and minimum pairwise distances of $P$.} of $P$, and $\dim$ is the {\em doubling dimension} of the dataset (an intrinsic dimensionality, see Section~\ref{sec:pre:double}). However, Vamana still cannot find the exact NN when $q\notin P$.


\vgap

A more practical scenario between $\dis(q, v^*) = 0$ and $\dis(q, v^*) = \infty$ is the case when $\dis(q, v^*)$ is bounded by a small constant $\tau > 0$, as observed by \cite{hd16, pcc+23} in real-world datasets. Recently, Peng et al.~\cite{pcc+23} proposed the $\tau$-monotonic graph ($\tau$-MG) such that after each step in a greedy routing, the distance between the next visited node and query point $q$ is reduced by a $\tau$-additive factor. $\tau$-MG is constructed with $\V = P\setminus\set{p}$ for each $p\in P$ and another edge pruning rule by setting $r = \dis(p, u) - 3\tau$ (see Figure~\ref{fig:lune-final}c). The average search time of $\tau$-MG is $O(n^{1/m}(\ln n)^2)$ when $P$ is from a uniform distribution. But as the starting point may have a very large distance to $q$, the worst-case search time is still $O(n)$.

\extraspacing{\bf Motivation.} As summarized in Table~\ref{tab:method_comparison}, no existing algorithm can find the exact NN in poly-logarithmic time under the assumption that $\dis(q, v^*) \le \tau$, even when the dataset has a constant doubling dimension. Solving this problem will address an open question in the theoretical foundation of proximity graphs. Another question is whether we can combine the advantages of Vamana with $\tau$-MG and propose a unified framework that, when $\dis(q, v^*) \le \tau$, finds the exact NN, and when $\dis(q, v^*) > \tau$, returns an ANN, with search times consistently poly-logarithmic. In this paper, we will design a new PG framework that addresses these questions affirmatively.

\vgap

From a practical standpoint, our goal is to enhance both the accuracy and efficiency of existing PG algorithms. We seek to develop practical PGs that provide robust query performance on real datasets, regardless of whether the condition $\dis(q, v^*) \le \tau$ holds.

\subsection{Our Contributions} \label{sec:intro:our-result}

\noindent{\bf A new proximity graph with enhanced guarantees.} We propose a new PG, the {\em $\alpha$-convergent graph} ($\alpha$-CG), in which we set $\V = P\setminus\set {p}$ for each $p\in P$, and then \REV{utilize a carefully designed edge pruning rule that incorporates both parameters $\tau$ and $\alpha$ to eliminate unnecessary candidates.} Specifically, as illustrated in Figure~\ref{fig:lune-final}d, we set $r = \fr{1}{\alpha}(\dis(p, u) - (\alpha+1)\tau)$, which reduces the radius of the ball $B(u, \dis(p, u))$ by first subtracting $(\alpha+1)\tau$ from $\dis(p, u)$ and then scaling a multiplicative factor of $\alpha$. The intersection area is strictly smaller than that of the MRNG and Vamana.

\vgap

With this new edge pruning rule, we prove that the $\alpha$-CG admits non-trivial theoretical guarantees. When $\dis(q, v^*) \le \tau$, our $\alpha$-CG can find the exact NN of $q$ in $O((\alpha\tau)^\dim \log \Delta \log_\alpha \Delta) $ time, thereby achieving the first algorithm that can find the exact NN in poly-logarithmic time when the doubling dimension is constant (see Table~\ref{tab:method_comparison}). A critical property of our PG is that after each hop, the distance from the next visited node to $q$ is reduced by an $\alpha$-multiplicative factor, and the search algorithm terminates in $O(\log_\alpha \Delta)$ steps. While $\tau$-MG and the slow preprocessing version of Vamana terminates in $O(n)$ and $O(\log_\alpha \fr{\Delta}{(\alpha-1)\eps})$ steps, respectively. When $\dis(q, v^*) > \tau$, a greedy routing on $\alpha$-CG can find an $(\fr{\alpha+1}{\alpha-1}+\eps)$-ANN of $q$ for any $\eps > 0$, with the same query accuracy as that of Vamana. The space, query time, and construction time of $\alpha$-CG match those of Vamana \cite{ix23}, up to an $O(\tau^\dim)$-factor (see Theorem~\ref{thm:qry-time}).


\extraspacing{\bf A practical variant with efficient construction.} To reduce the index construction time, similar to existing approaches, we propose a practical variant of our $\alpha$-CG, the {\em $\alpha$-convergent neighborhood graph} ($\alpha$-CNG). The graph is constructed by generating a {\em local-neighbor set} $\V$ for each data point $p$ (much smaller than $P\setminus \set{p}$), and subsequently applying our edge pruning rule to select a small subset of $\V$ as the out-neighbors of $p$. 

\vgap

Empirical analysis shows that the parameter $\alpha$ greatly affects the performance of our graph. Increasing $\alpha$ reduces the number of search steps but results in higher out-degrees in the PG. This trade-off complicates the selection of an optimal $\alpha$, a problem that subsequent works \cite{sds+19, ssks21,gks+23,jkgsa22} on Vamana have overlooked. We propose an {\em adaptive local pruning} strategy that adjusts $\alpha$ for each node locally during graph construction. Starting from a small value, we gradually increase $\alpha$ and prune candidates until the node’s out-degree reaches a predefined threshold, preserving long-distance shortcut edges and maintaining graph connectivity. To construct our graph efficiently, we implement a {\em distance-reusing} mechanism that stores and reuses intermediate computation results to accelerate the adaptive pruning process, along with a {\em lazy pruning} strategy that significantly reduces the number of pruning operations.

\extraspacing{\bf Experiments.} We compared our $\alpha$-CNG with 4 state-of-the-art PG indexes on 8 real-world datasets. At the same accuracy levels, $\alpha$-CNG usually reduced distance computations by at least 15\% and search steps by over 45\% when compared with the best-performing baseline, while the maximum speedups in distance computations and search steps can be 2.28x and 2.88x, respectively. These improvements indidate the faster convergence of our method, making $\alpha$-CNG particularly suitable for disk-based or distributed deployments, where I/O operations scale proportionally with search steps. Besides, both the index sizes and construction times of $\alpha$-CNG are comparable to existing PGs (e.g., HNSW). We also validated our edge pruning rule by integrating it into HNSW and Vamana, and the results show that it improved efficiency for most datasets, suggesting the effectiveness of our pruning rule.

\section{Preliminaries}
\subsection{Problem Setting and Basic Notations} \label{sec:pre:prob-def}

A {\em metric space} $(\X, \dis)$ consists of a set $\X$ of points and a distance function $\dis: \X \times \X \rightarrow \mathbb{R}_{\geq 0}$ satisfying the \REV{{\em triangle inequality}} $\dis(p_1, p_2) + \dis(p_2, p_3) \ge \dis(p_1, p_3), \: \forall p_1, p_2, p_3\in \X$. Let $P \subseteq \X$ be a set of $n$ data points from $\X$. \REV{Given a query point $q$, a point $v^*\in P$ is a {\em nearest neighbor} (NN) of $q$ if $\dis(v^*, q)\le \dis(p, q)$ for all $p\in P$; while a point $p'\in P$ is a {\em $c$-approximate nearest neighbor} ($c$-ANN) of $q$ for some constant $c > 1$ if $\dis(p', q)\le c\cdot \dis(p, q)$ for all $p\in P$.} We want to preprocess $P$ into a data structure with a small space that can answer exact NN or ANN queries efficiently. 

\vgap

For any set $S\subseteq \X$, the diameter of $S$ --- denoted as $\diam(S)$ --- is the maximum distance of two points in $S$, while the {\em aspect ratio of $S$} is the ratio between $\diam(S)$ and the minimum pairwise distance in $S$. We will use $\Delta$ to denote the aspect ratio of the input set $P$. In this paper, we will assume that the minimum pairwise distance in $P$ is exactly 1, as can be achieved by scaling the distance function $\dis$ appropriately. Hence, we have $\Delta = diam(P)$.

\vgap

This paper studies a practical scenario where the distance from the query point $q$ to its NN $v^*$ is bounded by a small constant $\tau \in (0, \Delta]$. Our objective is to develop a PG that can find the exact NN efficiently. For simplicity in the analysis, we assume $\tau \ge 1$. The scenario where $\tau < 1$ can be handled by using a solution designed for $\tau \ge 1$. We will explore how to support ANN queries efficiently when the condition $\mathrm{d}(q, v^*) \le \tau$ does not hold.

\vgap

Many real-world applications often require retrieving the top-$k$ nearest neighbors.  This leads to the \emph{approximate $k$-nearest neighbor search} problem, where each query returns a set of $k$ data points whose distances to $q$ are no further than the other data points by at most a constant factor. To evaluate the accuracy of a method, empirical studies often rely on ranking-based metrics that compare the returned set with the true top-$k$ results. A widely used metric is \emph{recall}, which measures the average fraction of the true $k$ nearest neighbors returned by the data structure.
\begin{algorithm}[!t]
	\caption{\textbf{beam-search($G$, $q$, $s$, $L$, $k$)}}
	\label{alg:greedy-search-impl}
	\begin{flushleft}
		\textbf{Input:} graph $G$, query point $q$, entry point $s$, queue size $L$ \\
		\textbf{Output:} $k$ ANN of $q$
	\end{flushleft}
	\begin{algorithmic}[1]
		\State candidate queue $Q\leftarrow \set{s}$
		\State explored set $\E = \emptyset$
		\While{$Q\setminus \E \ne \emptyset$}
		\State $u^* \leftarrow \arg\min \{ \dis(x, q) \mid x \in Q\setminus \E \}$
		\For{each out-neighbor $v$ of $u^*$}
		\State $Q\leftarrow Q \cup \set{v}$
		\EndFor
		\State keep the $L$ entries in $Q$ that are closest to $q$
        \State $\E \leftarrow \E \cup\set{u^*}$
		\EndWhile
		\State \textbf{return} $k$ points in $Q$ closest to $q$
	\end{algorithmic}
\end{algorithm}
\subsection{Proximity Graphs} \label{sec:pre:pg}

A {\em proximity graph} (PG) on $P$ is a simple directed graph $G$ whose vertices are precisely the points of $P$. We denote directed edges as $(u, v)$, representing arcs from vertex $u$ to vertex $v$, and define $N^+(u)$ as the set of out-neighbors of $u$ in $G$.

\vgap

Although various PG methods may use different strategies for connecting the edges in $G$, a simple greedy algorithm is commonly used to answer $k$-ANN queries. As shown in Algorithm~\ref{alg:greedy-search-impl}, given a query point $q$, an entry (starting) point $s\in P$, and a queue size $L$, the {\em beam search} algorithm maintains a queue $Q$ containing up to $L$ of the closest points visited during search. We say a point $u$ in $Q$ is {\em explored} if the distances of its out-neighbors to $q$ are computed. Initially, $Q$ contains only the entry point $s$. At each step, the algorithm selects a point $u^*\in Q$ that is closest to $q$ and has not been explored. It then visits the out-neighbors of $u^*$, attempts to insert them into $Q$, and marks $u^*$ as explored. We refer to $u^*$ as a {\em hop vertex}. When all the points in $Q$ are visited, the algorithm terminates by returning the $k$ points in $Q$ with the smallest distances to $q$. A special case is when $L = 1$, the algorithm --- referred to as {\em greedy grouting} --- traverses a sequence of hop vertices having descending distances to $q$, returning the last hop vertex as the answer.

\vgap

Following the convention of previous works \cite{ix23, fxwc19, pcc+23}, this paper will assume that each distance computation consumes constant time. The query time of the {beam search} algorithm is asymptotically bounded by the total number of distance computations, which is the total out-degree of the visited hop vertices.

\subsection{Doubling Dimension} \label{sec:pre:double}

Given a point $p$ and a real value $r > 0$, we will use $B(p, r)$ to represent the set containing all the points in $\X$ whose distance to $p$ is at most $r$; we refer to $B(p, r)$ as {\em a ball} centered at $p$ with radius $r$. To analyze the performance of PG-based algorithms, we introduce the notion of doubling dimension, which is often used to measure the “intrinsic dimensionality” of high-dimensional point sets \cite{gkl03, bkl06, kl04}:
\begin{definition}[Doubling Dimension]
	Let $(\mathcal{X}, \dis)$ be a metric space. A finite dataset $D \subseteq \mathcal{X}$ is said to have \emph{doubling constant} $\lambda$ if, for any point $p \in D$ and radius $r > 0$, the set $D \cap B(p, r)$ can be covered by at most $\lambda$ balls of radius $r/2$, and $\lambda$ is the smallest number with this property. The {doubling dimension} of $D$ is defined as $\log_2 \lambda$.
\end{definition}

The doubling dimension generalizes the Euclidean dimension. For any set $D\subseteq \real^m$, the doubling dimension of $D$ is $O(m)$.\footnote{\url{https://en.wikipedia.org/wiki/Doubling_space}} For example, when $m = 2$, any ball of radius $r$ can be covered by 7 balls with radius $r/2$, meaning that the doubling constant of $D$ is at most $\log_2 7$. Moreover, empirical studies \cite{fk94, ix23} showed that the doubling dimension of real data sets is often smaller than their ambient dimension, e.g., dimension $d$ of the Euclidean space. Following the previous studies on doubling dimension \cite{cg06, hm06, kl04}, this paper assumes that the doubling dimension of input $P$ --- denoted as $\dim$ --- is $O(1)$, namely, $P$ has a low doubling dimension. Based on the definition of doubling dimension, we can derive \cite{ix23}:
\begin{lemma} \label{lmm:cover-ball}
	\label{lem:doubling-dimension}
	Consider any set $P$ of points with doubling dimension $\dim$. For any point $p \in P$ and real values $R \ge r > 0$ satisfying $R/r = O(1)$, the set $P \cap B(p, R)$ can be covered by $O((R/r)^\dim)$ balls of radius $r$. Formally, there exist $p_1, \dots, p_s \in P$ such that:
	$$P \cap B(p, R) \subseteq \bigcup_{i=1}^s B(p_i, r) \textrm{  and  } s = O((R/r)^\dim).$$
\end{lemma}

\section{Alpha-Convergent Graph} \label{sec:alpha-pg}

\REV{This section presents a new PG, called the {\em $\alpha$-convergent graph ($\alpha$-CG)}, which offers stronger theoretical guarantees on query time and accuracy. }Similar to Vamana and $\tau$-MG, its construction time is $\Omega(n^2)$. We will provide a practical variant in Section~\ref{sec:practical-pg} to reduce the construction time.

\subsection{A New Pruning Strategy} \label{sec:alpha-pg:pruning-rule}

Recall that in PG construction, for each point $p\in P$, we need to select an appropriate subset from a candidate $\V$ as the out-neighbors of $p$. We consider the following sub-problem:

\begin{definition}[Neighbor Selection \cite{aep25}]\label{def:candidate-sel}
    For each point $p\in P$ and a candidate set $\V\subseteq P\setminus\set{p}$ of points, select a subset of $\V$ as the out-neighbors of $p$, while maintaining graph connectivity.
\end{definition}
    
An algorithm for this subproblem serves as a critical component in PG construction. The goal is to choose a ``shortcut" subset from $\V$ such that if $p$ is not the exact NN of $q$, then $p$ connects to an out-neighbor $v$ much closer to $q$ than $p$. To achieve the purpose, we introduce the following edge pruning rule:


\begin{definition}[\REV{Edge pruning rule of $\alpha$-CG}]\label{def:edge-pruning}
	Given any point $p$ and its candidate set $\V$, a point $u\in \V$ is pruned if there exists an out-neighbor $v$ of $p$ satisfying the following condition:
	\myeqn{
		\dis(p, u) > \alpha \cdot \dis(u, v) + (\alpha + 1) \cdot \tau \label{eqn:pruning-rule-ine}
	}
	where $\alpha > 1$ is a parameter.
\end{definition}
It can be verified \eqref{eqn:pruning-rule-ine} is equivalent to $v\in B(p, \dis(p,u)) \cap B(u, r)$ where $r = \fr{1}{\alpha}(\dis(p, u) - (\alpha+1)\tau)$, as shown in Figure~\ref{fig:lune-final}d. Given this pruning rule, we define a {\em pruning} procedure (Algorithm~\ref{alg:pruning-pro}) to solve the problem in Definition~\ref{def:candidate-sel}. Specifically, we sort the points in $\V$ in ascending order according to their distances to $p$. Then, we iteratively select a subset $S$ of $\V$ in the following manner. Initially, set $S = \emptyset$. For every point $u$ in this sorted sequence, check if there exists a point $v\in S$ satisfying the condition \eqref{eqn:pruning-rule-ine}. If not, add $u$ into $S$; otherwise, skip $u$ and proceed to check the next point. When all points in $\V$ are checked, the procedure returns $S$, which will be referred to as the {\em shortcut set of $p$} on $\V$. The next subsection will show that when $\V = P\setminus\set{p}$, the edges from $p$ to nodes in $S$ satisfy the ``shortcutable'' property with a fast convergence rate.

\begin{algorithm}[!t]
	\caption{\textbf{pruning($p$, $\V$, $\alpha$)}}
	\label{alg:pruning-pro}
	\begin{flushleft}
		\textbf{Input:} point $p$, candidate set $\V$, parameter $\alpha$ \\
		\textbf{Output:} a shortcut set $S$ of $p$ on $\V$
	\end{flushleft}
	\begin{algorithmic}[1]
		\State sort $\V$ in the ascending order of the distances to $p$
		\State shortcut set $S = \emptyset$
		\For{each $u\in \V$ (sorted)}
		\State {\bf if} there exists a point $v\in S$ satisfying inequality~\eqref{eqn:pruning-rule-ine}
		\State \> \>\>\>  {\bf then} continue
		\State {\bf else} $S \leftarrow S \cup \set{u}$
		\EndFor
		\State \textbf{return} $S$
	\end{algorithmic}
\end{algorithm}

\subsection{The $\alpha$-Convergent Graph} \label{sec:alpha-pg:graph}

\REV{Based on our new edge pruning rule, we define the $\alpha$-CG $G$ of $P$ as follows:}
\begin{definition} [$\alpha$-convergent graph]
	Every point of $P$ is a vertex of $G$ and vice versa. For each point $p\in P$, run the {pruning} procedure (Algorithm~\ref{alg:pruning-pro}) with $\V = P\setminus \set{p}$, and define the returned shortcut set $S$ as the out-neighbors of $p$ in $G$.
\end{definition}


\begin{figure}[t!]
  \centering

  \begin{subfigure}[t]{0.27\linewidth}
    \centering
    \includegraphics[width=0.80\linewidth]{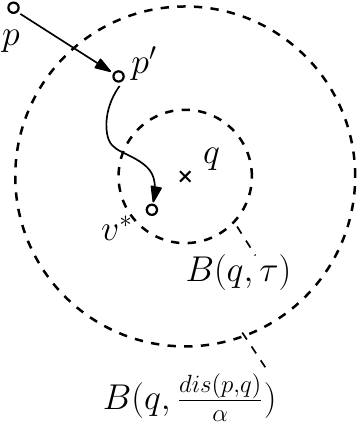}
    \caption{$\alpha$-reducible property}
    \label{fig:alpha_reducible}
  \end{subfigure}
  \hfill
  \begin{subfigure}[t]{0.30\linewidth}
    \centering
    \includegraphics[width=0.85\linewidth]{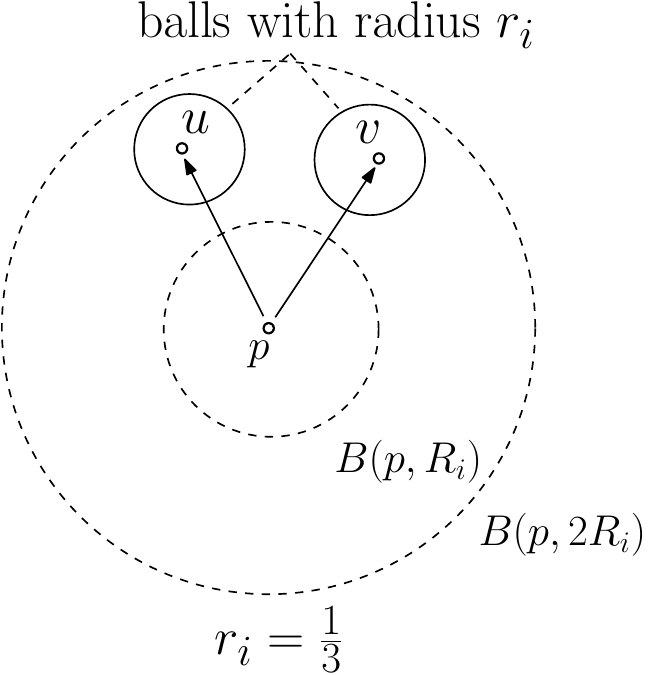}
    \caption{Out-degree analysis: Case 1}
    \label{fig:case1}
  \end{subfigure}
  \hfill
  \begin{subfigure}[t]{0.30\linewidth}
    \centering
    \includegraphics[width=0.8\linewidth]{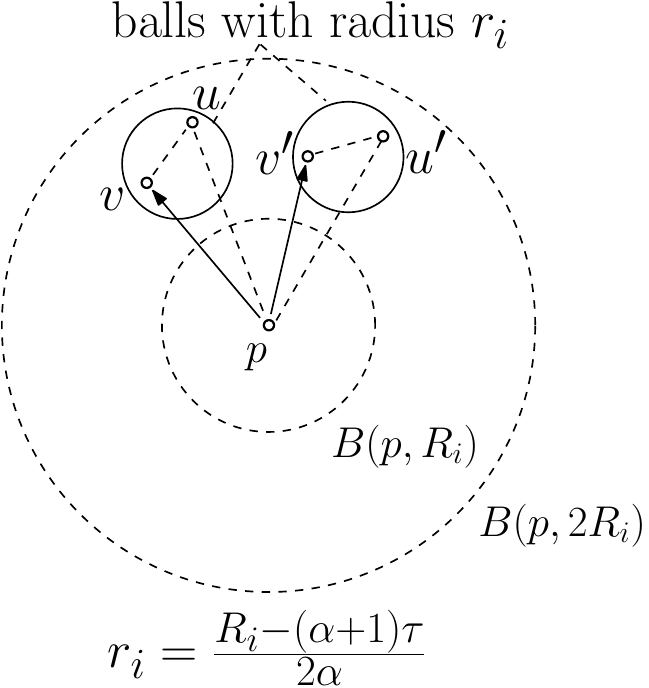}
    \caption{Out-degree analysis: Case 2}
    \label{fig:case2}
  \end{subfigure}

  \figcapup
  \caption{
  \REV{Geometric illustrations of the pruning rule.
  (a) Illustration of the $\alpha$-reducible property.
  (b,c) Two cases in the out-degree analysis.}
  }
  \figcapdown
  \label{fig:geometry_analysis}
\end{figure}

We begin by proving the following property. 
\begin{lemma}[$\alpha$-reducible property] \label{lmm:reducible}
	Consider any point $q$ whose exact NN $v^* \in P$ satisfies $\dis(q, v^*) \le \tau$. For every point $p \in P$ with $p \ne v^*$, either (i) $v^*$ is an out-neighbor of $p$ in $G$, or (ii) there exists an edge $(p, p')$ in $G$ such that $\dis(p', q) \le \dis(p, q)/\alpha$.
\end{lemma}
\begin{proof}
    Consider any $p \in P$ such that $p \neq v^*$ and $v^*$ is not an out-neighbor of $p$ in $G$. According to Definition~\ref{def:edge-pruning}, as $v^*$ is pruned from $p$'s candidate set, there must exist an out-neighbor $p'$ of $p$ such that (by setting $u = v^*$ and $v = p'$):
	\myeqn{
		\dis(p, v^*) > \alpha \cdot \dis(p', v^*) + (\alpha + 1) \cdot \tau. \label{eqn:lmm-condition1}
	}
	Combining the above with the \REV{triangle inequality} $\dis(p, q) \ge \dis(p, v^*) - \dis(q, v^*)$, we have
	\myeqn{
		\dis(p, q) &\ge& \dis(p, v^*) - \dis(q, v^*) \nn \\
		\textrm{(by \eqref{eqn:lmm-condition1})} &>& \alpha \cdot \dis(p', v^*) + (\alpha + 1) \cdot \tau -\dis(q, v^*) \nn \\
		&\ge& \alpha \cdot \dis(p', v^*) + \alpha \cdot \tau \label{eqn:lmm-condition2}
	}
	where the last inequality used the assumption that $\dis(q, v^*)\le \tau$. On the other hand, \REV{according to the triangle inequality}, we have
	\myeqn{
		\dis(p', q) \le \dis(p', v^*) +\dis(v^*, q) \le \dis(p', v^*) + \tau \label{eqn:lmm-condition3}.
	}
	Inequalities~\eqref{eqn:lmm-condition2} and~\eqref{eqn:lmm-condition3} together imply that $\dis(p, q)/\dis(p', q) \ge \alpha$, finishing the proof of the lemma.
\end{proof}

See an illustration of the property in Figure~\ref{fig:alpha_reducible}. Given any query $q$ satisfying $\dis(q, v^*)\le \tau$, 
Lemma~\ref{lmm:reducible} implies that greedy routing from any starting vertex quickly converges to the exact NN of $q$: at each step, if the current vertex is not $v^*$, then either the next hop reaches $v^*$ or the distance from the next vertex to $q$ decreases by an $\alpha$ factor.

\subsection{Theoretical Analysis} \label{sec:alpha-pg:analysis}

This subsection will establish:

\begin{theorem} \label{thm:qry-time}
    Let $P$ be a set of $n$ points and $G$ be the $\alpha$-CG of $P$. Consider any query point $q$ whose exact NN in $P$ is $v^*$. 
    \myitems{
        \item If $\dis(q, v^*)\le \tau$, a greedy routing on $G$ can find $v^*$ in $O((\alpha\cdot \tau)^\dim \cdot \log\Delta \log_\alpha \Delta)$ time. \item Otherwise, a greedy routing on $G$ can find an $(\fr{\alpha+1}{\alpha-1}+\eps)$-ANN  of $q$ in $O((\alpha\cdot \tau)^\dim \cdot \log\Delta \log_\alpha \fr{\Delta}{(\alpha-1)\eps})$ time, for any $\eps > 0$. 
    
    }
    The $\alpha$-CG $G$ has space $O(n\cdot (\alpha\tau)^\dim \log \Delta)$ and can be constructed in $O(n^2\cdot ((\alpha\tau)^\dim \log \Delta+ \log n))$ time.
\end{theorem}

The above result shows that our $\alpha$-CG captures the Vamana graph. When $\tau = 1$, our $\alpha$-CG achieves the same space and query time while providing a slightly better accuracy guarantee (it can find the exact NN when $\dis(q, v^*)\le 1$). For $\tau > 1$, our method further improves query accuracy with only a modest increase in space and search time: $\alpha$-CG can find the exact NN of $q$ when $\mathrm{d}(q, v^*) \le \tau$, with both space and query time incurring an additional $O(\tau^\dim)$ factor. Readers may observe the exponential dependence on $\alpha$ and $\tau$; these are theoretical upper bounds, and in our experiments, small values of $\alpha$ and $\tau$ already yield strong empirical performance.

\vgap

The rest of the subsection serves as a proof of Theorem~\ref{thm:qry-time}. We first show that the maximum out-degree of $G$ is $O((\alpha\tau)^\dim \log \Delta)$ and analyze its construction time. After that, we analyze the search path length of the greedy routing algorithm under the two cases $\dis(q, v^*) \le \tau$ and $\dis(q, v^*) > \tau$.

\subsubsection{\bf Space and Construction Time} We first prove that $G$ has a low maximum out-degree:
\begin{lemma} \label{lmm:out-degree}
	The $\alpha$-CG $G$ of $P$ has maximum out-degree $O((\alpha \cdot \tau)^\dim\log \Delta)$, where $\dim$ is the doubling dimension of $P$.
\end{lemma}
\begin{proof}
Consider any point $p\in P$. The out-neighbors of $p$ in $G$ is the returned set of {pruning}($p$, $P\setminus\set{p}$, $\alpha$), which is the shortcut set $S$ of $p$ on $P\setminus{p}$. For each integer $i \in [0, \lc \log_2 \Delta \rc]$, define: 
\myeqn{
    \text{Ring}_i = \left\{ p' \in P \mid R_i < \dis(p, p') \leq 2R_i \right\},  \text{ where } R_i = \frac{\Delta}{2^{i+1}}. \label{eqn:ring}
}
Recall that $\Delta$ is the diameter of $P$, and the minimum pairwise distance of points in $P$ is 1. For every $p'\in P$, because $\dis(p, p') \in [1, \Delta]$, there is a unique $i \in [0, \lc \log_2 \Delta \rc]$ such that $p'\in \text{Ring}_i$. Hence, we partition $P\setminus \set{p}$ into $\lc \log_2 \Delta \rc + 1$ subsets. 

\vgap

Consider any $i \in [0, \lc \log_2 \Delta \rc]$. We will prove that $p$ has $O((\alpha \cdot \tau)^\dim)$ out-neighbors in $\text{Ring}_i$. Because there are $O(\log \Delta)$ rings, the out-degree of $p$ is thus $O((\alpha \cdot \tau)^\dim \log \Delta)$.

\vgap

\textbf{Case 1:} $R_i \le 2(\alpha+1)\tau$. In this case, we use balls of radius \( r_i=\frac{1}{3} \) to cover the ball $B(p, 2R_i)$ and hence $\text{Ring}_i$. According to Lemma~\ref{lem:doubling-dimension}, the number of such balls is 
\myeqn{
    O\left(\left(\fr{2R_i}{1/3}\right)^\dim\right) &=& O\left(\left(\fr{4(\alpha+1)\tau}{1/3}\right)^\dim\right) \nn \\
    \text{(by $\alpha > 1$)} &=& O\left(\left(\fr{8\alpha \cdot \tau}{1/3}\right)^\dim\right)
    = O((\alpha \cdot \tau)^\dim) \nn
}
where the last equality used the assumption $\dim = O(1)$.
Since the minimum distance between any two points in $P$ is \( 1 \), each ball of radius $\fr{1}{3}$ contains at most one point (see Figure~\ref{fig:case1}). Thus, $|\text{Ring}_i| = O((\alpha \cdot \tau)^\dim)$, meaning that $p$ has $O((\alpha \cdot \tau)^\dim)$ out-neighbors in $\text{Ring}_i$.


\vgap

\textbf{Case 2:} $R_i > 2(\alpha+1)\tau$. We cover $\text{Ring}_i$ using balls of radius $r_{i}= \frac{R_i - (\alpha + 1)\tau}{2\alpha}$. For any two points $u, v\in \text{Ring}_i$, if $u$ and $v$ are in the same ball with radius $r_i$, then, we have:
\myeqn{
    \dis(u, v) &\leq&  2r_i = \frac{R_i - (\alpha + 1)\tau}{\alpha} \notag \\
\Rightarrow \dis(u, v) \cdot \alpha &\le& R_i - (\alpha+1)\tau < \dis(p, u) - (\alpha+1)\tau \nn \\
\Rightarrow \dis(p, u) &>& \dis(u, v) \cdot \alpha + (\alpha + 1)\tau. \nn
}
Similarly, we can obtain $\dis(p, v) > \dis(u, v) \cdot \alpha + (\alpha + 1)\tau$. According to the pruning rule defined by \eqref{eqn:pruning-rule-ine}, at most one of $u$ and $v$ can be in the shortcut set $S$ of $p$. Hence, for any ball $B$ with radius $r_i$, at most one point in $\text{Ring}_i \cap B$ can be an out-neighbor of $p$. For instance, in Figure~\ref{fig:case2}, when $v$ and $v'$ are the our-neighbors of $p$, then edges $(p, u)$ and $(p, u')$ will not exist in our PG.

\vgap

Next, we will show that $B(p, 2R_i)$, and hence $\text{Ring}_i$, can be covered by $O(\alpha^\dim)$ balls with radius $r_i$, implying $p$ has $O(\alpha^\dim)$ out-neighbors in $\text{Ring}_i$. By Lemma~\ref{lem:doubling-dimension}, the number of balls with radius $r_i$ needed to cover $\text{Ring}_i$ is at most:
\myeqn{
    O\left( \left( \frac{2R_i}{r_i} \right)^{\dim} \right) &=& O\left( \left( 2R_i \cdot \fr{2\alpha}{R_i - (\alpha+1)\tau} \right)^{\dim} \right) \nn \\
    &=& O \left( 4\alpha\left( 1 + \fr{(\alpha+1)\tau}{R_i - (\alpha+1)\tau}\right)^\dim\right) \nn \\
    \text{(by $R_i \ge 2(\alpha+1)\tau$)}&=& O ( \left( 8\alpha \right)^\dim ) = O(\alpha^\dim) \nn
}
where the last inequality used the assumption $\dim = O(1)$.

\vgap

Combining the above two cases and the fact $\tau \ge 1$, we conclude that the maximum out-degree of $G$ is $O((\alpha \cdot \tau)^\dim \log \Delta)$.
\end{proof}

\noindent{\bf Construction time.} The {pruning}($p$, $\V$, $\alpha$) procedure (Algorithm~\ref{alg:pruning-pro}) finishes in $O(|\V|\cdot (\log |\V|+|S|))$ time. To see this, Line 1 of algorithm~\ref{alg:pruning-pro} finishes in $O(|\V|\log |\V|)$ time. For each loop of Line 3, we need to check if a point $u\in \V$ can be pruned by a point in $S$. A linear scan finishes in $O(|S|)$ time. Because there are $|\V|$ iterations, the total time of the algorithm is thus $O(|\V|\cdot (\log |\V|+|S|))$.

\vgap

Recall that the out-neighbor set $N^+(p)$ of each point $p\in P$ is obtained by calling {pruning}($p$, $\V$, $\alpha$) with $\V = P \setminus \set{p}$. According to Lemma~\ref{lmm:out-degree}, we have $|S| = |N^+(p)| = O((\alpha\cdot \tau)^\dim \log \Delta)$, implying that the {pruning} procedure computes the out-neighbors of $p$ in $O(n \cdot (|N^+(p)| + \log n)) = O(n \cdot ((\alpha\cdot \tau)^\dim \log \Delta + \log n))$ time. As $|P| = n$, we can conclude that the total construction time of the $\alpha$-CG is $O(n^2\cdot ((\alpha\cdot \tau)^\dim \log \Delta + \log n))$.

\subsubsection{\bf Query Time} 

\noindent{\bf When $\dis(q, v^*) \le \tau$.} The $\alpha$-reducible property has an important implication on the behavior of the {\bf beam search} algorithm, even when $L = 1$: If the current hop vertex $p$ visited by the algorithm is not the exact NN $v^*$, $p$ has an out-neighbor (thus the next hop vertex) whose distance to $q$ is reduced by an $\alpha$-multiplicative factor. Based on this property, we have the following result about the number of hop vertices visited:
\begin{lemma} \label{lmm:greedy-steps}
	Consider any query point $q$ whose exact NN $v^*$ in $P$ satisfies $\dis(q, v^*) \le \tau$. The greedy routing (beam search with $L = 1$) algorithm starting with any entry point $s$ in $G$ can find the exact NN of $q$ by visiting $O(\log_\alpha \Delta)$ hop vertices. 
\end{lemma}
\begin{proof}
	Let $\sigma = (p_1, p_2, ..., p_\ell)$ be the sequence of $\ell$ hop vertices visited by greedy grouting.  We thus have $p_1 = s$. To see $p_\ell = v^*$, suppose the last visited hop vertex $p_\ell \ne v^*$. Then, according to Lemma~\ref{lmm:reducible}, there exists an out-neighbor of $p_\ell$ in $G$ that is closer to $q$ than $p_\ell$, contradicting that the algorithm terminates at $p_\ell$.
	
	\vgap
	
	Next, we prove $\ell = O(\log_\alpha \Delta)$. For any entry point $s\in P$, $\dis(s, q) \le \dis(s, v^*) + \dis(v^*, q) \le \diam(P) + \tau \le 2\Delta$. Recall that $\tau \le \Delta$ and $\diam(P)$ is the diameter of the dataset $P$; as the minimum pairwise distance in $P$ is 1, we have $\Delta = \diam(P)$ (see Section~\ref{sec:pre:prob-def}). 
    
	
	\vgap
	
	
	

Our first claim is that for every $i \in [1, \ell - 2]$, we have $\dis(p_{i+1}, q) \le \dis(p_i, q)/\alpha$. To see this, observe that $v^*$ cannot be an out-neighbor of $p_i$; otherwise, the greedy grouting would terminate at $p_{i+1}$, visiting $i+1 \le \ell - 1$ hop vertices and thereby contradicting the fact that $\sigma$ has size $\ell$. Then, by Lemma~\ref{lmm:reducible}, we have $\dis(p_{i+1}, q) \le \dis(p_i, q)/\alpha$, and 
\myeqn{
		\dis(p_{\ell-1}, q) &\le& \dis(p_1, q)/\alpha^{\ell-2} \nn \\
		\Rightarrow \ell &\le& \log_\alpha (\dis(p_1, q)/\dis(p_{\ell-1}, q)) + 2. \label{eqn:dis-converge}
}


Now, we prove that $\dis(p_{\ell-1},q) \ge 1/2$. If $\dis(v^*,q) \ge 1/2$, then, since $v^*$ is the exact NN of $q$ in $P$, we have $\dis(p_{\ell-1},q) > \dis(q,v^*) > 1/2$. On the other hand, if $\dis(v^*,q) < 1/2$, by \REV{triangle inequality} and the fact that the minimum pairwise distance in $P$ is $1$, we obtain
$$
\dis(p_{\ell-1}, q) \ge \dis(p_{\ell-1}, v^*) - \dis(q, v^*) \ge 1 - 1/2 = 1/2.
$$
Hence, in both cases $\dis(p_{\ell-1}, q)\ge 1/2$. Plugging this bound into Inequality~\eqref{eqn:dis-converge}, we obtain
\myeqn{
\ell \le \log_\alpha (2\Delta/(1/2)) + 2 = O(\log_\alpha \Delta).\nn
}
Thus, the algorithm finds the exact NN of $q$ after visiting $O(\log_\alpha \Delta)$ hop vertices, and the lemma follows.
\end{proof}

According to Lemma~\ref{lmm:out-degree}, every vertex of $G$ has a maximum out-degree at most $O((\alpha\cdot \tau)^\dim \log \Delta)$. Since a greedy routing can find the exact NN of $q$ by visiting $O(\log_\alpha \Delta)$ hop vertices (Lemma~\ref{lmm:greedy-steps}), the search algorithm finishes in $O((\alpha\cdot \tau)^\dim \log \Delta \log_\alpha \Delta)$ time. This proves the first bullet of Theorem~\ref{thm:qry-time}.

\extraspacing{\bf When $\dis(q, v^*) > \tau$.} We say a proximity graph $G'$ is {\em $\alpha$-shortcut reachable} if for every two vertices $p, z$ of $G'$ such that $(p, z)$ is not in $G'$, then, there exists an edge $(p, p')$ in $G'$ such that $\dis(p', z) \le \dis(p, z)/\alpha$. Indyk and Xu \cite[Theorem 3.4]{ix23} proved that given an $\alpha$-shortcut reachable PG, a greedy routing starting at any vertex can answer an $(\fr{\alpha+1}{\alpha-1}+\eps)$-ANN query after visiting $O(\log_\alpha \fr{\Delta}{(\alpha-1)\eps})$ hop vertices. Our $\alpha$-CG $G$ is also $\alpha$-shortcut reachable:
\myitems{
    \item For any two vertices $p, z$ such that $(p, z)$ is not in $G$, by definition of the pruning rule (Definition~\ref{def:edge-pruning}), there must exist a vertex $p'$ such that $\dis(p', z) < \fr{1}{\alpha}(\dis(p, z) - (\alpha+1)\tau) < \fr{1}{\alpha}\dis(p, z)$, implying that $G$ is $\alpha$-shortcut reachable.
}
Together with Lemma~\ref{lmm:out-degree}, we can conclude that the query time is $O((\alpha\tau)^\dim \log \Delta \log_\alpha \fr{\Delta}{(\alpha-1)\eps})$. 
The second bullet of Theorem~\ref{thm:qry-time} then follows. This finishes the proof of Theorem~\ref{thm:qry-time}.

\extraspacing{\bf Remark.} The main difference among our $\alpha$-CG and existing PGs, such as MRNG, $\tau$-MG, and Vamana, lies in the edge pruning rules, which significantly impact query performance. This paper finds a crafted edge pruning rule that leads to a new PG with enhanced guarantees, supported by a non-trivial theoretical analysis.

\section{Alpha-Convergent Neighborhood Graph with Efficient Construction}\label{sec:practical-pg}

Despite the superior asymptotic query performance of the $\alpha$-CG, its construction time is $\Omega(n^2)$ in the worst case. In line with existing works such as NSG \cite{fxwc19}, Vamana \cite{sds+19}, and $\tau$-MNG \cite{pcc+23}, we propose a practical variant that approximates $\alpha$-CG. This section also explores strategies for adaptively pruning candidates (setting the parameter $\alpha$) and constructing the index efficiently.

\begin{algorithm}[!t]
  \caption{\textbf{$\alpha$-CNG-construction}}
  \label{alg:construct-cpg}
  \begin{flushleft}
    \textbf{Input:} 
    data points $P$, 
    $K$ for $K$-NN graph, 
    out-degree threshold $M$,
    queue size $L$,
    candidate size $C$, 
    initial parameter $\alpha_0$, 
    max parameter $\alpha_{\max}$, 
    step size $\Delta\alpha$,
    runing parameter $\tau$ \\
    \textbf{Output:} $\alpha$-CNG $G$
  \end{flushleft}
  \begin{algorithmic}[1]
  

\Statex /* Phase 1: approximate $K$-NN graph and entry point */

    \State $G_0 \leftarrow$ \textbf{BuildApproxKNNGraph}$(P, K)$
    \State $s \leftarrow$ \textbf{beam-search}$(G_0, \text{centroid}(P), \text{random\_point}, L, k=1)$

\Statex /* Phase 2: candidate generation and pruning */

    \State $G\leftarrow$ a graph with vertex set $P$ and no edges
    \For{each point $p \in P$} 
      \State $\V \leftarrow$ \textbf{beam-search}$(G_0, p, s, L, C)$
      \State $N^+(p) \leftarrow$ \textbf{adaptive-pruning}$(p, \V, M, \alpha_0, \alpha_{\max}, \Delta\alpha)$
    \EndFor

\Statex /* Phase 3: backward edge insertion and lazy pruning */

    \State {\bf for} each $(u, v)$ in $G$, insert $(v, u)$ into $G$
    \For{each $p$ with $|N^+(p)| > M$}
      \State $N^+(p) \leftarrow$ \textbf{adaptive-pruning}$(p, N^+(p), M, \alpha_0, \alpha_{\max}, \Delta\alpha)$
    \EndFor

\Statex /* Phase 4: connectivity examination */


    \State DFS-tree $T \leftarrow$ \textbf{DFS}$(G, s)$ 
    \State {\bf for} each $p \in P \setminus T$, add necessary edges from $T$ to $p$

    \State \textbf{return} $G$
  \end{algorithmic}
\end{algorithm}

\subsection{The Overall Framework} \label{sec:practical-pg:nb}

We employ a standard framework utilized in existing PG methods, such as NSG \cite{fxwc19} and $\tau$-MNG \cite{pcc+23}, to generate a \emph{local-neighbor set} (as opposed to $P \setminus {p}$) as the candidate set for each point $p \in P$. Then, we apply an adaptive pruning strategy (Section~\ref{sec:practical-pg:alpha}) to select up to $M$ points as the out-neighbors of $p$, where $M$ is a predefined maximum out-degree that ensures efficient graph storage and computation. We refer to resulting graph $G$ as the \emph{$\alpha$-convergent neighborhood graph} ($\alpha$-CNG). This framework consists of the following phases, as summarized in Algorithm~\ref{alg:construct-cpg}.

\vgap

The initial phase constructs an approximate $K$-NN graph $G_0$, where each data point is connected to its approximate $K$-NNs in $P$. Efficient implementations are available for constructing such graphs, e.g., \cite{hd16, dcl11}. We then identify the {\em navigating node} $s$ by performing a beam search, starting from a random point with the centroid of $P$ (the geometric mean of all points in $P$) as the query point. The returned NN is designated as $s$.

\vgap

Next, we generate a {\em local-neighbor set} $\V$ for each point $p$ using $G_0$. We execute a beam search starting from node $s$, using $p$ as the query. All points visited during this search (whose distances to $p$ are computed) are gathered. Let $\V$ be the set of the $C$ closest points to $p$ that are collected. We then apply our {adaptive pruning} method to select a subset of at most $M$ points from $\V$, which will be assigned as the out-neighbors of $p$ in the $\alpha$-CNG $G$.

\vgap

Subsequently, backward edges are added to the graph for each edge inserted in the previous phase, ensuring bidirectional connectivity. If any vertex $p$'s out-degree exceeds $M$, its out-neighbor set is pruned accordingly using our adaptive pruning method. Finally, we verify the graph connectivity by performing a depth-first search (DFS) starting from $s$. To ensure that all nodes are reachable from $s$, we add necessary edges for any nodes not included in the DFS tree, in line with the NSG method \cite{fxwc19}.


\begin{figure}[t!]
    \centering
    \includegraphics[height=18mm]{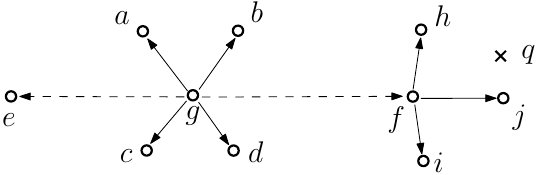}
    \figcapup
    \caption{When the shortcut set $S = \set{a,b,c,d,e,f}$ of $g$ has size larger than $M=4$, the two edges $(g, e)$ and $(g, f)$ are pruned.}
    \figcapdown
    \label{fig:large-alpha}
    \figcapdown
\end{figure}
\subsection{Adaptive Local Pruning} \label{sec:practical-pg:alpha}

To incorporate our pruning strategy (Algorithm~\ref{alg:pruning-pro}) into the above framework, a critical question is the parameter setting for $\alpha$, which greatly impacts the query performance of our $\alpha$-CNG. According to Lemma~\ref{lmm:reducible}, each time a new hop vertex is visited, the distance to $q$ decreases by at least a factor of $\alpha$, unless the routing directly reaches the exact NN $v^*$.
While a larger $\alpha$ accelerates convergence to the ANNs, it also leads to an exponential increase in the out-degree of each node (Lemma~\ref{lmm:out-degree}). Therefore, achieving a balance between out-degree and convergence rate is essential for efficient query performance. Despite the use of $\alpha$ in Vamana, a widely utilized PG, little attention has been paid to determining its optimal value.

\vgap

Recall that in our framework, each node can have an out-degree of at most $M$. However, the shortcut set $S$ returned by the {pruning} method (Algorithm~\ref{alg:pruning-pro}) may exceed this size, particularly when $\alpha$ is large. A common approach used in the literature is to return the $M$ closest points in $S$ to $p$, but this may lead to the omission of long-distance shortcut edges. For instance, as illustrated in Figure~\ref{fig:large-alpha}, the vertices from $a$ to $f$ form the shortcut set of $g$. When $M = 4$, both $e$ and $f$ become disconnected from $g$, resulting in the pruning of the longest shortcut edges $(g, e)$ and $(g, f)$ from the PG. Consequently, given the query point $q$ in Figure~\ref{fig:large-alpha}, a greedy routing will fail to find the exact NN $j$ of $q$ when starting from node $g$. This example illustrates that when the size of the shortcut set exceeds the threshold $M$, the longest shortcut edges are omitted from the constructed graph, thereby reducing graph connectivity.

\vgap

We propose an adaptive strategy (Algorithm~\ref{alg:adaptive-prune-wrapper}) that gradually adjusts the parameter $\alpha$ and preserves as many long-distance shortcut edges as possible. Since each data point has a different local-neighbor set, a global $\alpha$ may lead to varying shortcut set sizes across data points. Therefore, we tune the parameter $\alpha$ for each data point $p$ locally. Specifically, we start with a small initial value of $\alpha = \alpha_0$ and run the pruning method. Whenever the returned shortcut set $S$ is smaller than the threshold $M/2$, we increase $\alpha$ by a small step size $\Delta\alpha$ and rerun the pruning method to find a larger $S$. Note that the size of $S$ is highly sensitive to $\alpha$, so we opt to increase $\alpha$ by a small value each step. This procedure aims to find the first $\alpha$ for which the corresponding shortcut set $S$ exceeds the size of $M/2$. We then return the $M$ points in $S$ closest to $p$. This avoids selecting an excessively large $\alpha$ and $S$, which could lead to the pruning of long-distance shortcut edges. In the last iteration, the pruning method can terminate early once it has already collected $M$ points.

\vgap


\subsection{Efficient Graph Construction}

This subsection presents two optimizations to achieve efficient graph construction: (1) a distance reusing mechanism for adaptive edge pruning, and (2) a lazy pruning strategy for backward edge insertion. We also provide an analysis of the construction time.
\begin{algorithm}[!t]
	\caption{\textbf{adaptive-pruning}$(p, \V, M, \alpha_0, \alpha_{\max}, \Delta\alpha)$}
	\label{alg:adaptive-prune-wrapper}
	\begin{flushleft}
		\textbf{Input:} point $p$, candidate set $\V$, maximum out-degree threshold $M$, initial parameter $\alpha_0$, max parameter $\alpha_{\max}$, step size $\Delta \alpha$\\
		\textbf{Output:} a set of at most $M$ shortcut points
	\end{flushleft}
	\begin{algorithmic}[1]
		\State $\alpha \leftarrow \alpha_0$
            \State $S \leftarrow \emptyset$
		\While{$|S| < M/2$ and $\alpha \le \alpha_{\max}$}
		    \State $S \leftarrow$ \textbf{pruning}$(p, \V, \alpha)$
		    \State $\alpha \leftarrow \alpha + \Delta \alpha$
		\EndWhile
		\State \textbf{return} $M$ closest points of $S$ to $p$
	\end{algorithmic}
\end{algorithm}
\extraspacing{\bf A distance-reusing mechanism.} 
Although the adaptive pruning strategy offers a practical solution for tuning the parameter $\alpha$ and preserving long-distance shortcut edges, it introduces additional computational overhead. Recall that for each data point $p$ and its candidate set $\V$, the adaptive algorithm calls the {pruning} method (Algorithm~\ref{alg:pruning-pro}) multiple times for a sequence of increasing $\alpha$ values. As analyzed in Section~\ref{sec:alpha-pg:analysis}, the running time of the {pruning} method is $O(|\V|\cdot (\log |\V|+|S|))$. Let $\alpha_0, ..., \alpha_h$ be the sequence of tested $\alpha$ values and $S_0, ..., S_h$ be the corresponding shortcut sets returned. Since we only need to sort $\V$ once (Line 1 of algorithm~\ref{alg:pruning-pro}), the total time is $O(|\V| \cdot (\log |\V| + \sum_{i = 0}^h |S_i|) = O(|\V| \cdot (\log |\V| + M \cdot h))$, as $|S_i| \le M$ for each $i \in [0, h]$. It can be verified that the running time is dominated by the number of distance computations.

\vgap

To reduce the computational overhead, we reuse the intermediate computation results based on the following observation. When invoking {pruning}($p$, $\V$, $\alpha_i$) to obtain the shortcut set $S_i$, according to inequality~\eqref{eqn:pruning-rule-ine}, a point $u\in \V$ is pruned if and only if there exists a point $v\in S_i$ satisfying the condition $\dis(p,u)>\alpha_i\cdot\dis(u, v)+(\alpha_i+1)\cdot\tau$. Rearranging this leads to 
\myeqn{
    \frac{\dis(p, u) - \tau}{\dis(u, v) + \tau} > \alpha_i.
}
Define $\bar{\alpha}(u, v) = \frac{\dis(p, u) - \tau}{\dis(u, v) + \tau}$. Hence, a point $v\in S_i$ can prune $u$ if and only if $\bar{\alpha}(u, v) > \alpha_i$. When it is the first time to evaluate $\bar{\alpha}(u, v)$ for $\alpha = \alpha_i$, we can store and reuse it in subsequent iterations by simply comparing $\alpha(u, v)$ with $\alpha_j$ for $j > i$. This ensures that $\bar{\alpha}(u, v)$ is evaluated at most once for each pair $(u, v) \in \V\times \V$.

\vgap

Denote by $S^+ = \bigcup_{i = 0}^h S_i$, i.e., the set of all points in $\V$ that ever appeared in $S_i$ for some $i\in [0, h]$. As we only need to compute $\bar{\alpha}(u, v)$ for $v\in S^+$ and $u\in \V$. The total number of stored $(u, v)$ pairs and thus the number of distance computations is $O(|\V| \cdot |S^+|)$. Empirical analysis reveals that successive shortcut sets $S_i$ and $S_{i+1}$ exhibit significant overlap. The underlying reason is that as $\alpha$ increases, the pruning condition becomes more permissive, allowing most of the points in $S_i$ to remain in $S_{i+1}$, which implies that $|S^+|$ is close to $M$. Therefore, with the distance-reusing heuristic, the total number of distance computations is close to $O(|\V|\cdot M)$ in practice (rather than $O(|\V|\cdot M \cdot h)$), and the overhead introduced by calling the {pruning} method multiple times becomes limited.

\extraspacing{\bf Lazy pruning during backward edge insertion.} In the backward edge insertion phase (Section~\ref{sec:practical-pg:nb}), the conventional way used by $\tau$-MNG and NSG inserts each backward edge one by one. Whenever a node's out-degree exceeds $M$, a pruning procedure is invoked to enforce the out-degree constraint. We propose a lazy pruning strategy that invokes our adaptive pruning procedure at most once for each node. Specifically, for each point $p$, we first collect all backward edges starting from $p$ and merge them with $p$’s existing out-edges to form a candidate set $\V$. If $|\V| > M$, we invoke the {adaptive-pruning} procedure to select the final out-neighbors; otherwise, $\V$ is assigned directly as the new neighbor set of $p$.

\extraspacing{\bf Total construction time.} The construction time of $\alpha$-CNG is dominated by two parts: (i) approximate $K$-NN graph construction, and (ii) candidates pruning for all $p\in P$. We utilize the NN-decent algorithm \cite{dcl11} to compute the approximate $K$-NN graph, whose empirical time complexity is sub-quadratic.  

\vgap

Thanks to the lazy pruning strategy, we invoke the {adaptive-pruning} method at most twice for each $p\in P$ (once in phase 2 and once in phase 3). Since each candidate set $\V$ has a size at most $C$ ($C \le 500$ in our experiments), the running time for each adaptive pruning is $O(C \cdot (\log C + M\cdot h_{\max})$, where $h_{\max}$ is the maximum number of $\alpha$ values tested when pruning $\V$. Thus, the total running time is $O(n\cdot C(\log C + M\cdot h_{\max})) + f(n)$ where $f(n)$ is the running time of the NN-decent. We conclude that our construction time is sub-quadratic and comparable to existing practical PGs.

\subsection{Discussions} 


\noindent{\textbf{Parameter configuration.} Our experiments found that the shortcut set size may already reach $M/2$ when $\alpha < 1$ for certain data points. Therefore, the adaptive pruning method begins with $\alpha_0 = 0.9$; we set $\Delta\alpha = 0.05$ and $\alpha_{\max}=1.6$.

\vgap


Although $\tau$ is a problem parameter, in our experiments, we set $\tau$ to a small value within the range $[0, 10]$ (and in most cases, $\tau$ is smaller than 1) to avoid the large maximum out-degree. Both $\tau$-MNG \cite{pcc+23} and our paper found that the optimal settings for $\tau$ may vary across different datasets. Hence, we apply a grid search on the test queries. In practical applications, historical queries can be leveraged to determine this parameter. Specifically, we first compare the results obtained with $\tau=0$ against those with $\tau\in\{10, \, 1,\,0.1,\,0.01,\,0.001\}$ to identify a coarse range of $\tau$, and then fine-tune $\tau$ to locate the best value.



\vgap

\noindent{\bf Updates.} This paper focuses on the development of static PGs that enhance both query accuracy and query time. Supporting updates is beyond the scope of this work. Nonetheless, concepts from existing research—such as periodic global rebuilding \cite{pcc+23}, lazy deletion with masking \cite{ssks21}, and in-place update strategies \cite{xmb+25}—could be applied to our $\alpha$-CNG and may be considered for future work.

\section{Experimental Evaluation} \label{sec:exp}

Section~\ref{sec:exp:setup} describes the datasets and competing methods in our evaluation. Section~\ref{sec:exp:search} assesses the query performance of our approach compared to baselines, as well as its scalability on large datasets. Section~\ref{sec:exp:para-scala} analyzes the impact of the parameters $\tau$ and $\alpha$ in our methods. Section~\ref{sec:exp:construction} evaluates the index construction performance. Finally, Section~\ref{sec:exp:our-rule} explores the effectiveness of our edge pruning rule by replacing those used in other PGs with our own.

\begin{table}
\centering
\caption{Dataset statistics}
\label{tab:dataset_stats}
\figcapdown\figcapdown
\begin{tabularx}{0.7\linewidth}{lccccc}
\toprule
\textbf{Dataset} & \textbf{Dim.} & \textbf{\# Base} & \textbf{\# Query} & \textbf{Source} & \textbf{Type} \\
\midrule
\ttt{SIFT}    & 128  & 1M     & 10K  & \cite{jds10} & Image \\
\ttt{CRAWL}   & 300  & 1.98M  & 10K   & \cite{mgbpj18} & Text \\
\ttt{WIKI}    & 384  & 1M     & 1K    & \cite{mxbs16} & Text \\
\ttt{MSONG}   & 420  & 1M     & 200   & \cite{bewl11} & Audio \\
\ttt{LAION-I2I} & 768 & 1M     & 10K   & \cite{sbv+22} & Image \\
\ttt{GIST}    & 960  & 1M     & 1K    & \cite{jds10} & Image \\
\midrule
\ttt{DEEP100M}  & 96   & 100M & 10K & \cite{bl16} & Image \\
\ttt{BIGANN100M}  & 128  & 100M & 10K & \cite{swa+22} & Image \\

\bottomrule
\end{tabularx}
\end{table}

\subsection{Experiment Settings}\label{sec:exp:setup}

All experiments were conducted on a Linux server equipped with an Intel(R) Xeon(R) Gold 6430 CPU and 512 GB RAM, running Ubuntu 20.04. All methods were implemented in C++ and compiled with \texttt{g++} using the \texttt{-O3} optimization flag. 


\extraspacing{\bf Datasets.} We utilized eight real-world datasets, which are widely adopted in ANN search evaluation~\cite{pcc+23,ggxl25,cscz24,fwc21,fxwc19,my20,yxl+24}. These datasets include six at the 1M scale and two at the 100M scale. The 1M-scale datasets span various application domains, including image (\ttt{SIFT}, \ttt{LAION-I2I}, \ttt{GIST}), audio (\ttt{MSONG}), and text (\ttt{WIKI}, and \ttt{CRAWL}). We also assessed the scalability of our methods using the 100M-scale datasets, including \ttt{DEEP100M}\footnote{\ttt{DEEP100M} consists of the first 100 million vectors from the public \ttt{DEEP1B}\cite{bl16} dataset.} and \ttt{BIGANN100M}. Table~\ref{tab:dataset_stats} summarizes the key statistics, including the dimensionality (Dim.), the number of base points (\# Base), the number of query points (\# Query), data source, and data type.

\extraspacing{\bf Competing Methods.} Our first method, \texttt{$\alpha$-CNG}, employs an adaptive pruning strategy (Section~\ref{sec:practical-pg:alpha}). To evaluate the effectiveness of this strategy, we examined \texttt{Fixed-$\alpha$-CNG}, a variant that employs a global $\alpha$ for all data points. Since previous studies \cite{lzs+19,wxyw21} have consistently shown that PG-based methods 
outperform non-PG methods such as LSH and IVF, 
we focus our comparisons on PG-based baselines. We compared our methods against four state-of-the-art PG algorithms: \texttt{HNSW}~\cite{my20}, \texttt{Vamana}~\cite{sds+19}, \texttt{NSG}~\cite{fxwc19}, and \texttt{$\tau$-MNG}~\cite{pcc+23}, selected for their robust performance. Three additional popular algorithms—NSSG \cite{fwc21}, DPG \cite{lzs+19}, and FANNG \cite{hd16}—were excluded as $\tau$-MNG \cite{pcc+23} outperforms their performance. All baseline implementations use publicly available source code.


\vgap

The construction parameters for all structures were selected based on official recommendations and empirical evaluations. \texttt{HNSW} used a configuration of $M$ = 32 and $ef_C$ = 500 for all datasets, while \texttt{Vamana} used the settings in~\cite{sds+19}, with $M$ = 70, $L$ = 75, and $\alpha$ = 1.2. 

\vgap

\texttt{NSG}, \texttt{$\tau$-MNG}, \texttt{Fixed-$\alpha$-CNG}, and \texttt{$\alpha$-CNG} shared core graph parameters $K$, $M$, $L$, and $C$. We set $K$=200 for all datasets. For \ttt{SIFT} and \ttt{GIST}, we used the settings $M$=50, $L$=40, $C$=500 and $M$=70, $L$=60, $C$=500, respectively, following NSG repository recommendations. For \ttt{WIKI} and \ttt{LAION-I2I}, we applied a configuration of $M$=70, $L$=60, and $C$=500. For \ttt{MSONG}, and \ttt{CRAWL}, we set $M$=100, $L$=100, and $C$=500. For the two large-scale datasets \ttt{DEEP100M} and \ttt{BIGANN100M}, we adopted the configurations of $M$=100, $L$=100, $C$=500 and $M$=80, $L$=100, and $C$=500 respectively. The pruning-related parameters, $\alpha$ and $\tau$, were empirically tuned for optimal performance. For \texttt{$\alpha$-CNG}, we set $\alpha_0 = 0.9$ (except for \ttt{BIGANN100M}, where $\alpha_0 = 1$ was used to include more long edges), $\Delta\alpha = 0.05$, and $\alpha_{\max} = 1.6$.

\extraspacing{\bf Metrics.} Query accuracy was measured using recall@\(k\), defined as \(\text{recall@}k = \frac{|\mathrm{KNN}(q) \cap \mathrm{Res}(q)|}{k}\), where \(\mathrm{KNN}(q)\) denotes the exact \(k\)-NN of the query \(q\), and \(\mathrm{Res}(q)\) is the set of query results returned by the algorithm. Query efficiency was evaluated using two metrics: (1) the number of distance computations (NDC), which dominates the overall search cost and provides a platform-independent measure of efficiency; (2) the number of hops, defined as the number of hop vertices (search steps) visited by the search algorithm. For each dataset, we recorded the average recall@100, average NDC, and average number of hops of all provided queries.


\subsection{Search Performance} \label{sec:exp:search}
\begin{figure*}[t]  
  \centering
\includegraphics[width=0.8\textwidth]{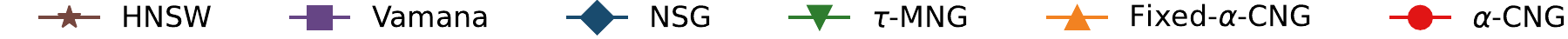}
\par\vspace{0pt}
  \begin{subfigure}[t]{0.32\textwidth}
      \centering
            \captionsetup{skip=-3pt}  
    \includegraphics[width=\linewidth]{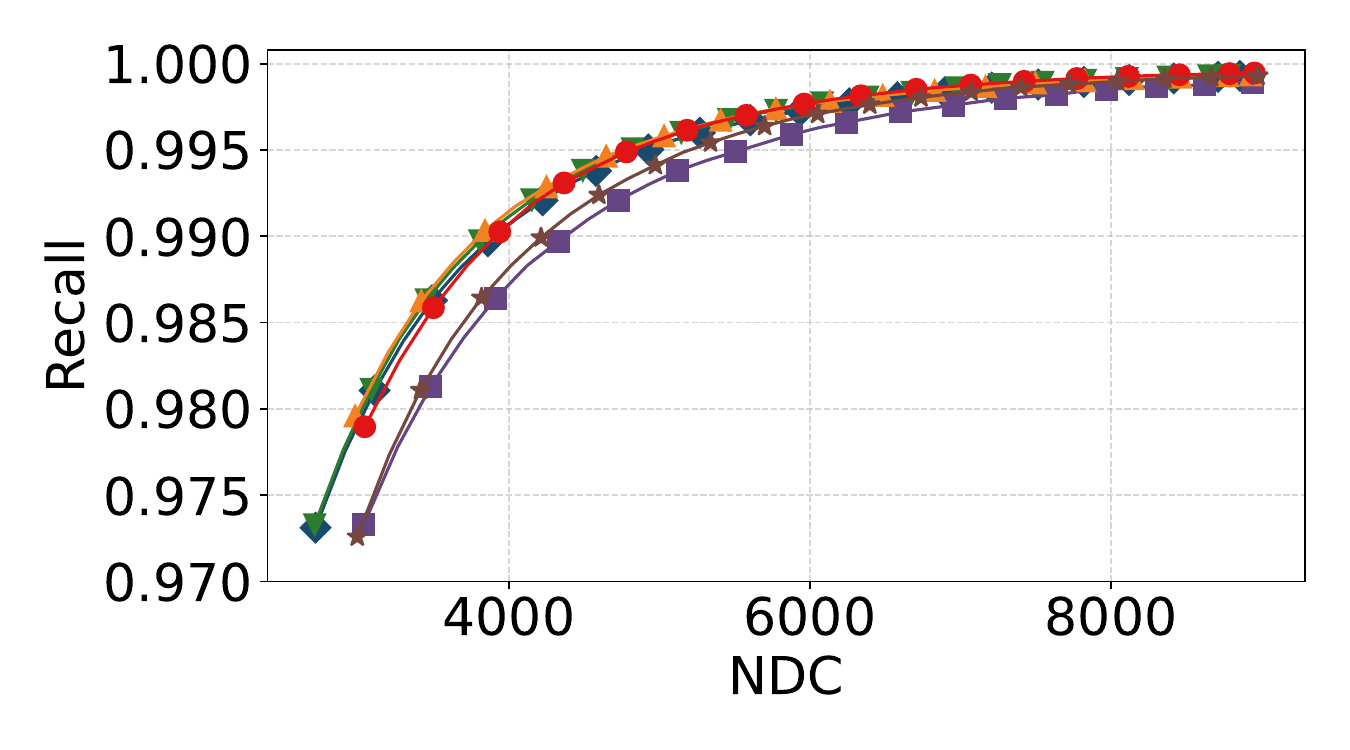}
    \caption{\ttt{SIFT}}
  \end{subfigure}
  \hfill
    \begin{subfigure}[t]{0.32\textwidth}
          \centering
      \captionsetup{skip=-3pt}  
    \includegraphics[width=\linewidth]{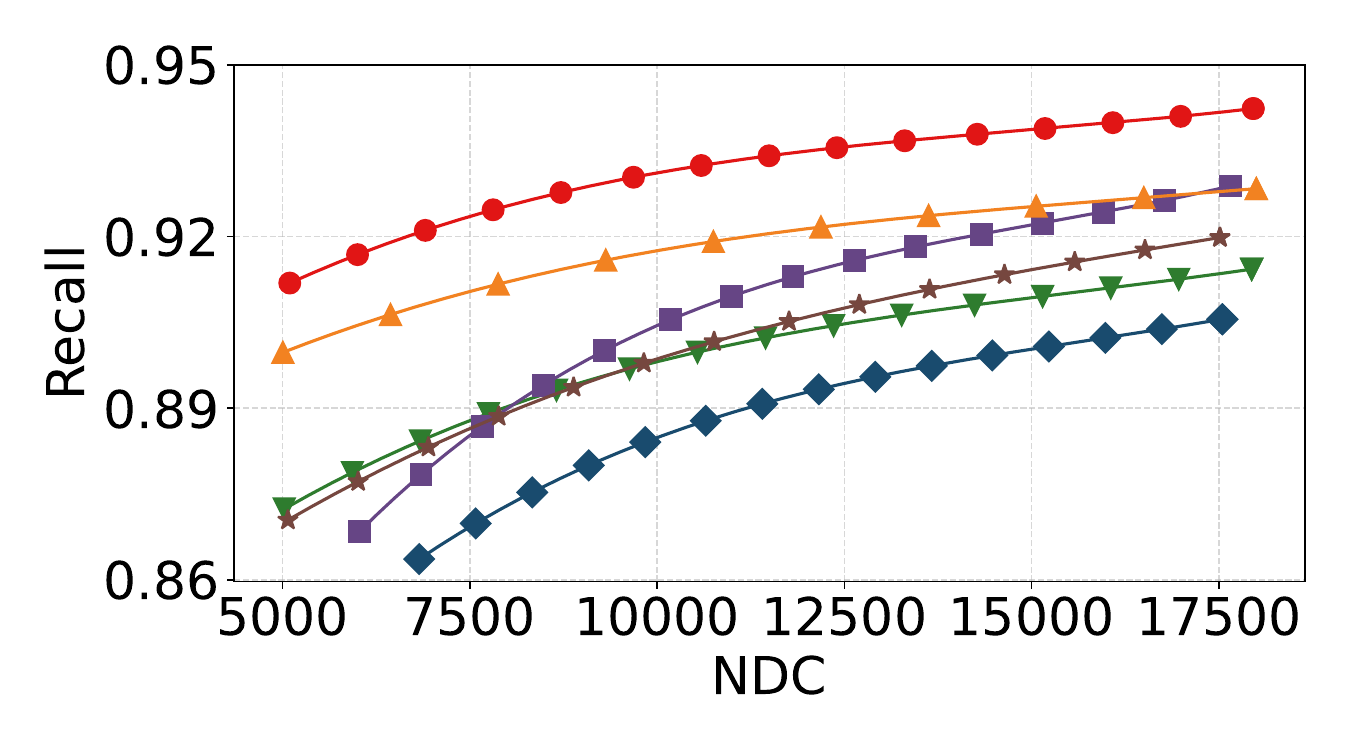}
    \caption{\ttt{CRAWL}}
  \end{subfigure}
  \hfill
    \begin{subfigure}[t]{0.32\textwidth}
          \centering
        \captionsetup{skip=-3pt}  
    \includegraphics[width=\linewidth]{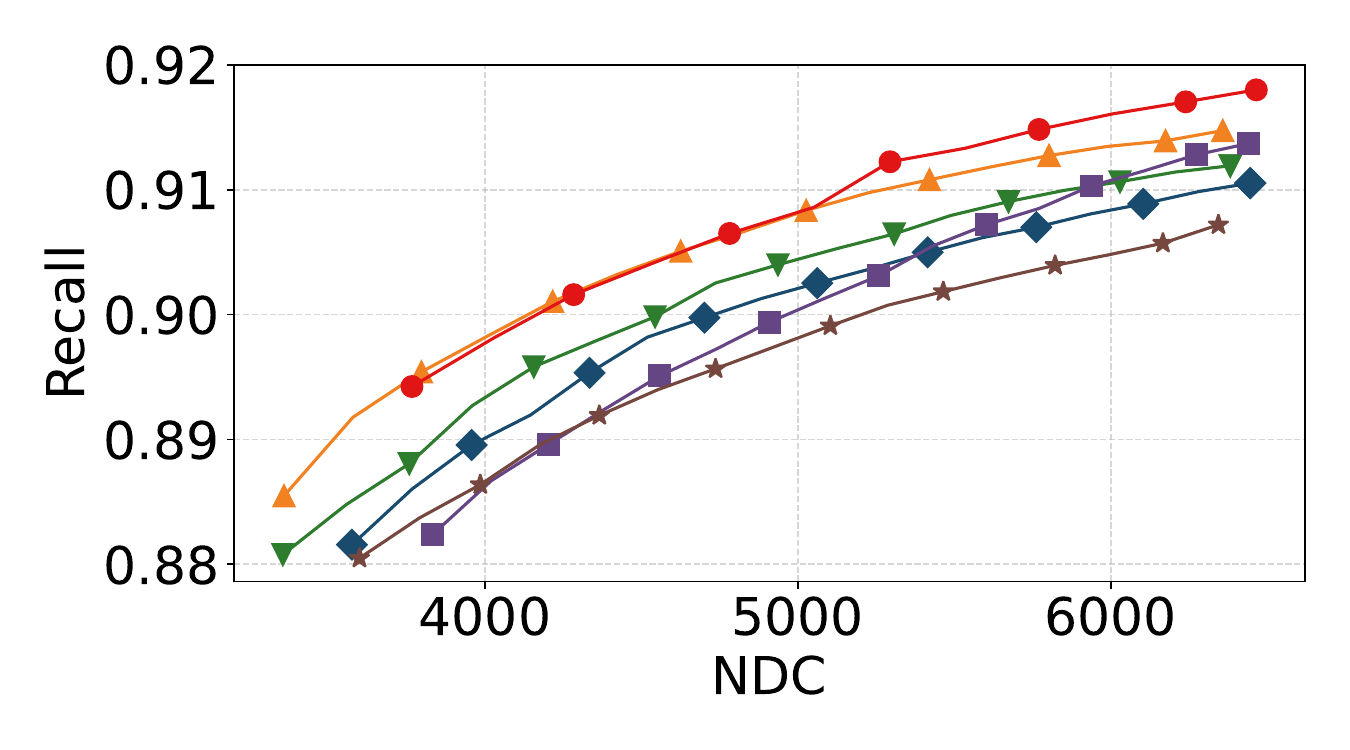}
    \caption{\ttt{WIKI}}
  \end{subfigure}
  \begin{subfigure}[t]{0.32\textwidth}
  \centering
        \captionsetup{skip=-3pt}  
    \includegraphics[width=\linewidth]{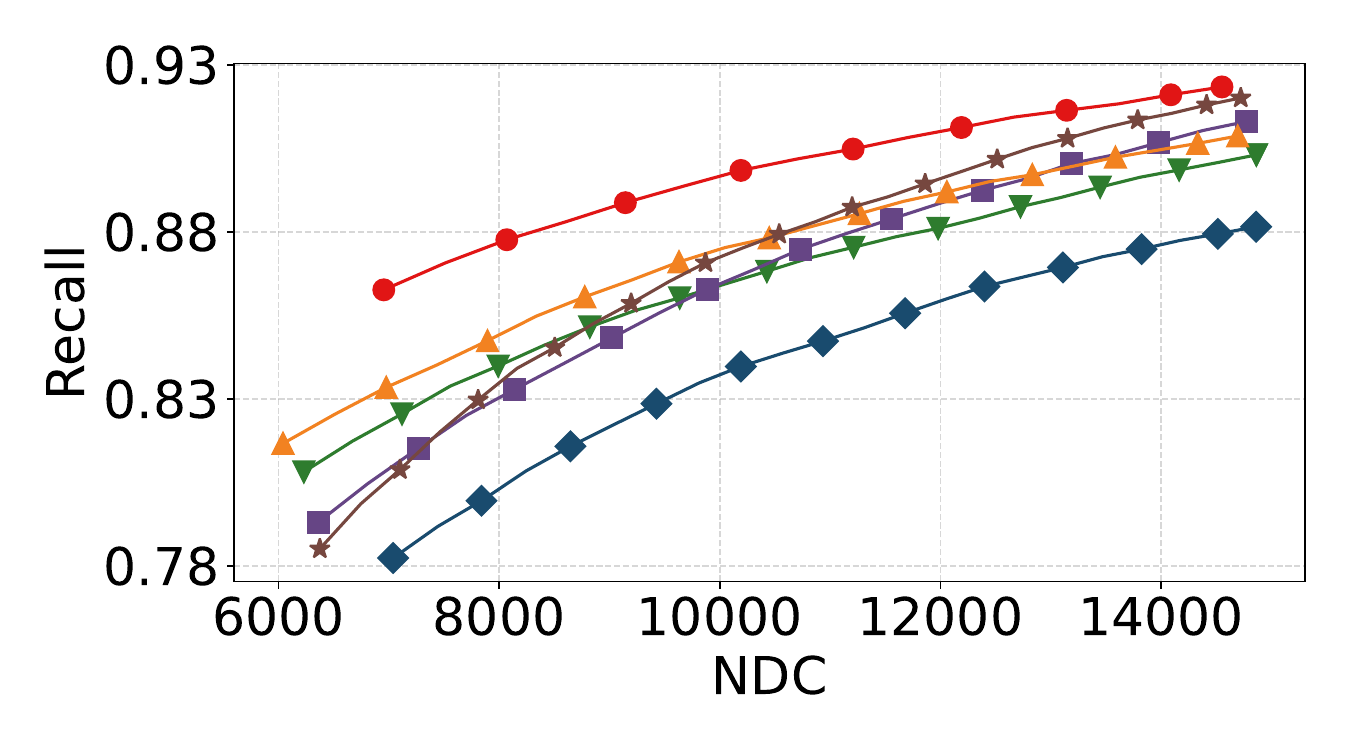}
    \caption{\ttt{MSONG}}
  \end{subfigure}
  \hfill
  \begin{subfigure}[t]{0.32\textwidth}
  \centering
        \captionsetup{skip=-3pt}  
    \includegraphics[width=\linewidth]{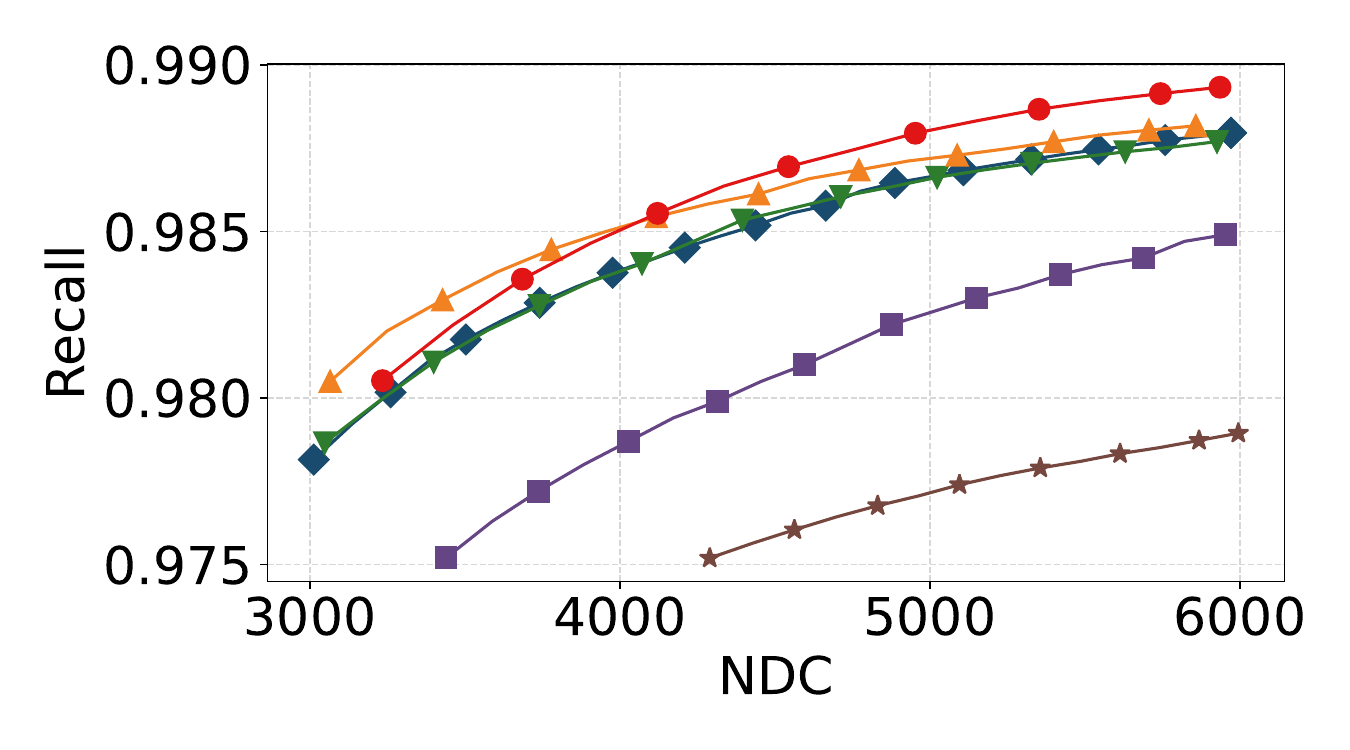}
    \caption{\ttt{LAION-I2I}}
  \end{subfigure}
  \hfill
    \begin{subfigure}[t]{0.32\textwidth}
    \centering
        \captionsetup{skip=-3pt}  
    \includegraphics[width=\linewidth]{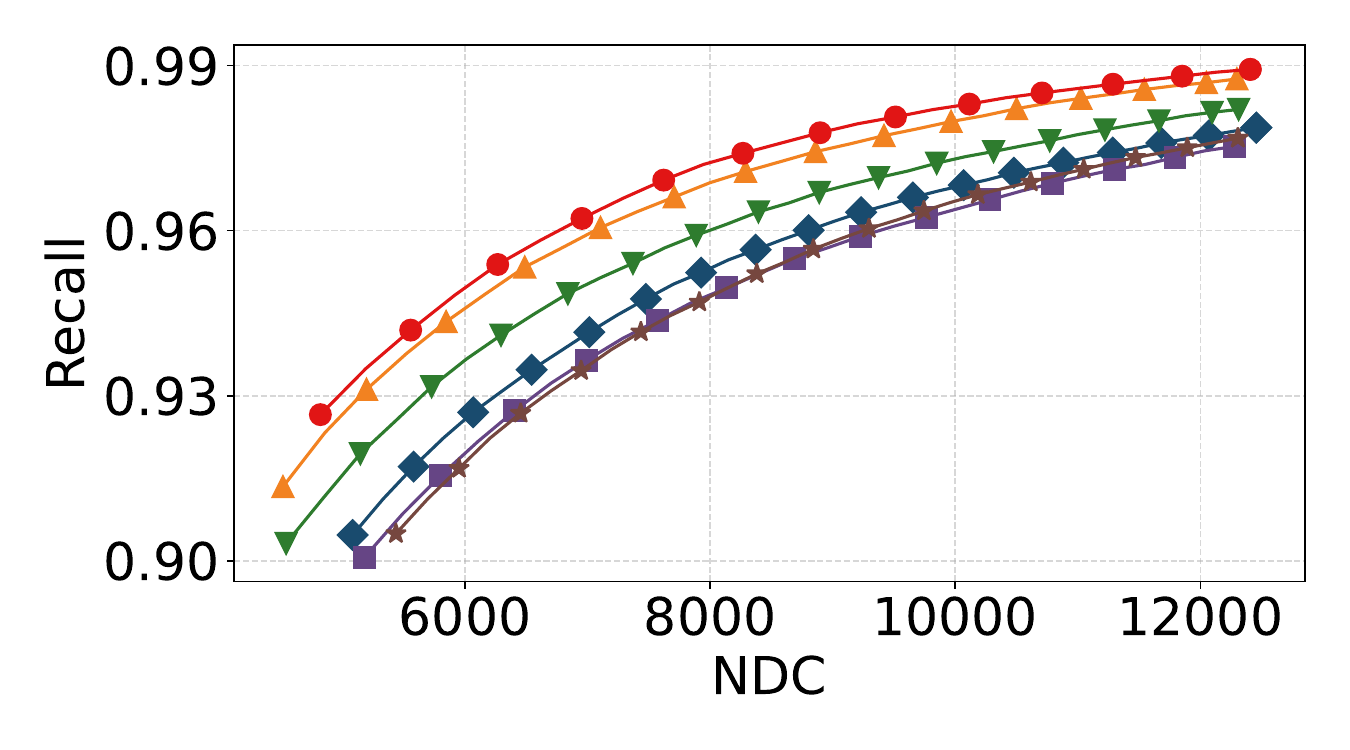}
    \caption{\ttt{GIST}}
  \end{subfigure}
    \caption{Recall@100 vs. NDC}
    \label{fig:NDC-vs-recall}
  \begin{subfigure}[t]{0.32\textwidth}
  \centering
        \captionsetup{skip=-3pt}  
    \includegraphics[width=\linewidth]{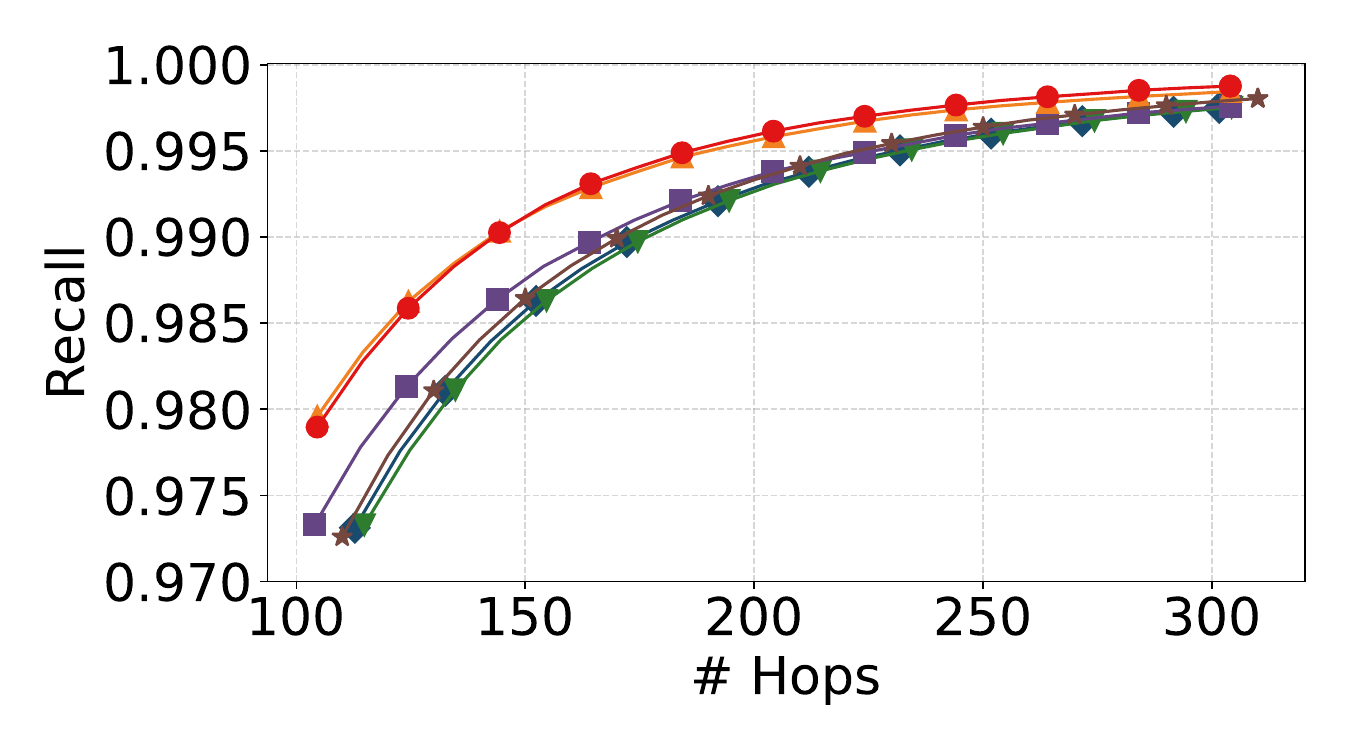}
    \caption{\ttt{SIFT}}
  \end{subfigure}
  \hfill
  \begin{subfigure}[t]{0.32\textwidth}
  \centering
        \captionsetup{skip=-3pt}  
    \includegraphics[width=\linewidth]{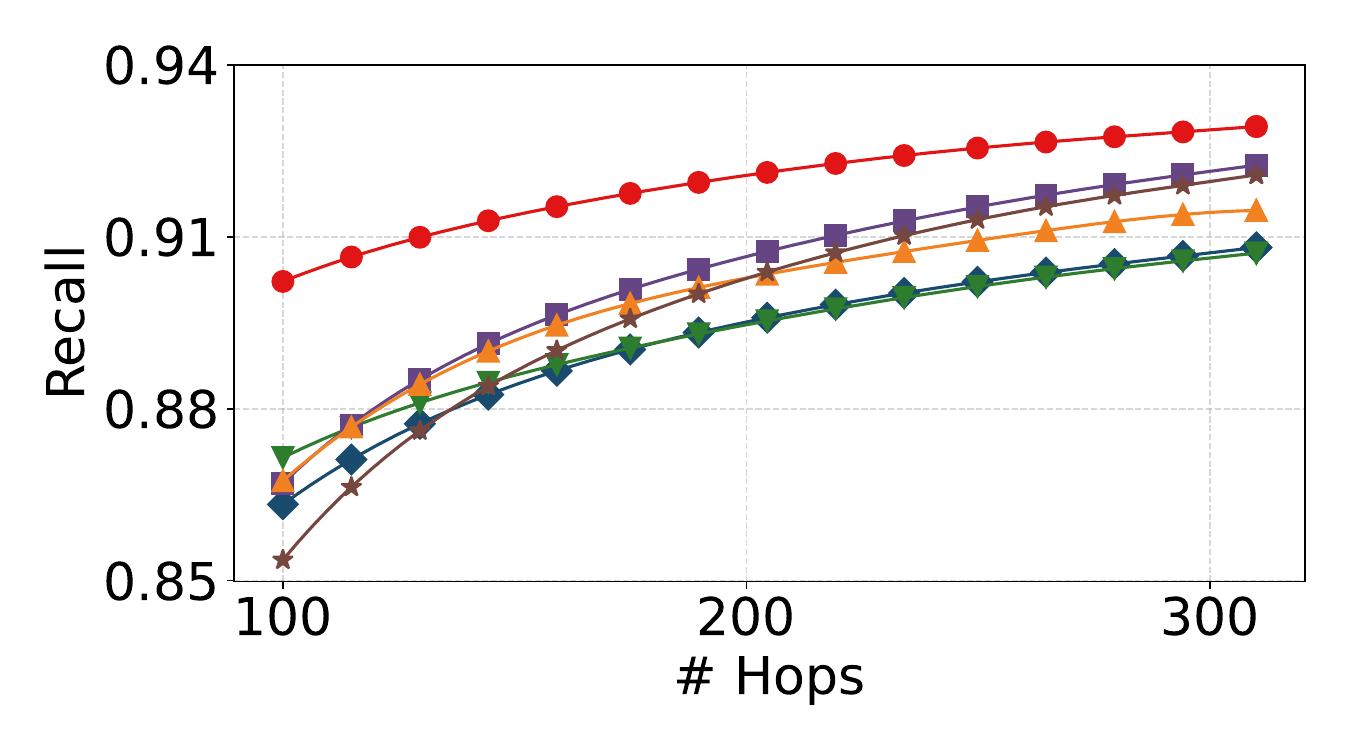}
    \caption{\ttt{CRAWL}}
  \end{subfigure}
  \hfill
    \begin{subfigure}[t]{0.32\textwidth}
      \centering
        \captionsetup{skip=-3pt}  
    \includegraphics[width=\linewidth]{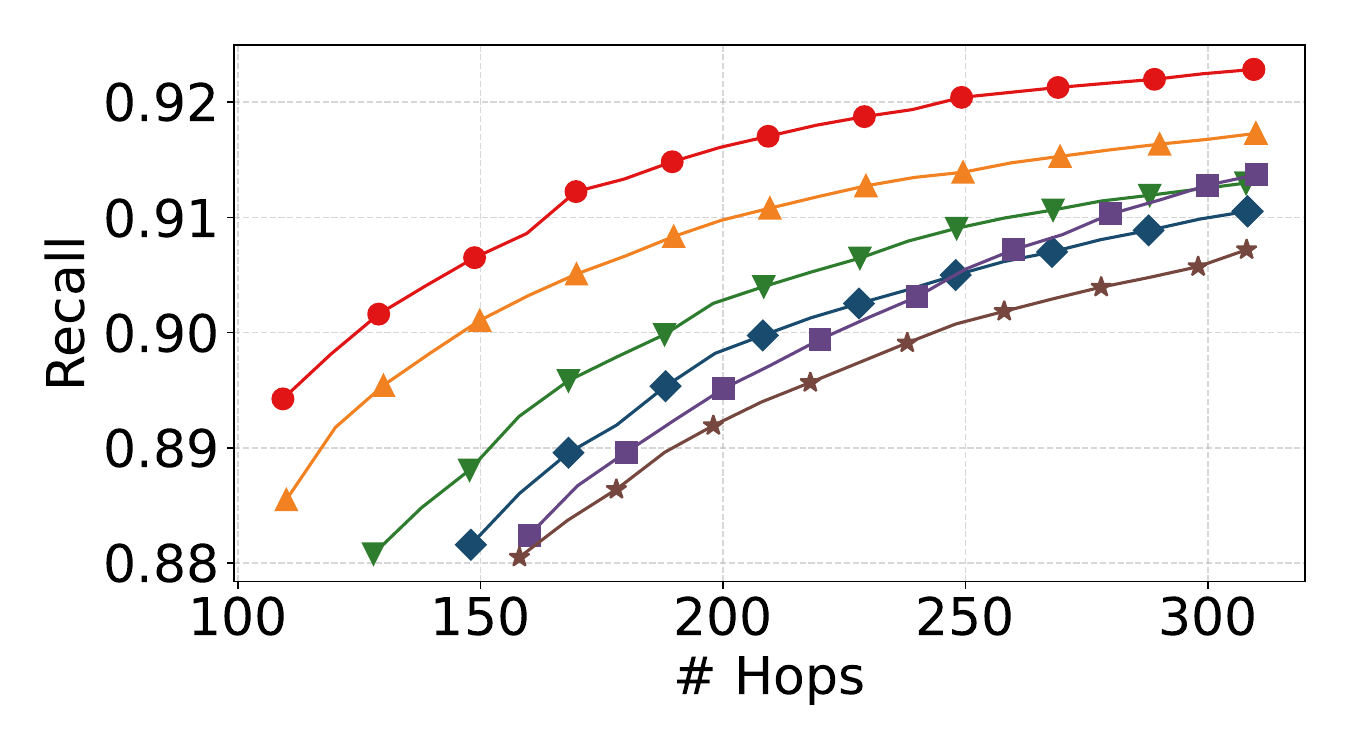}
    \caption{\ttt{WIKI}}
  \end{subfigure}
  \begin{subfigure}[t]{0.32\textwidth}
  \centering
        \captionsetup{skip=-3pt}  
    \includegraphics[width=\linewidth]{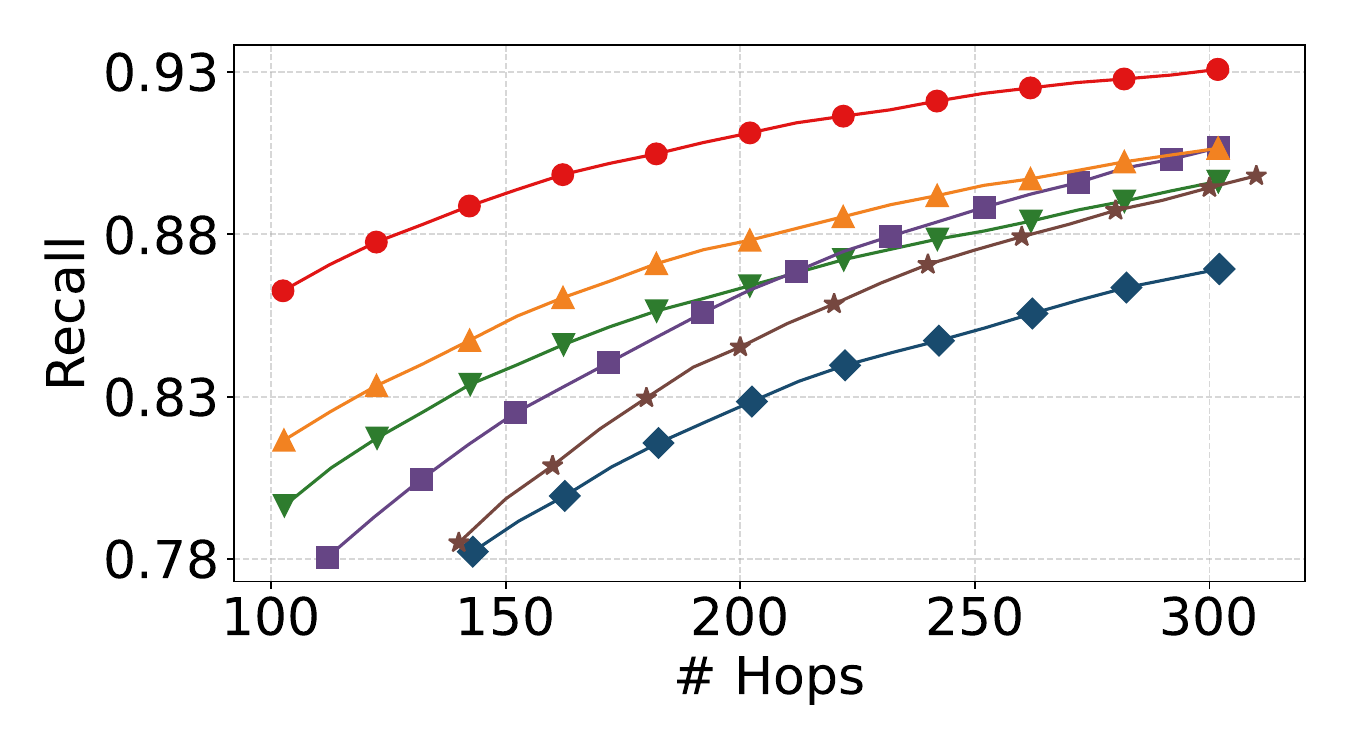}
    \caption{\ttt{MSONG}}
  \end{subfigure}
  \hfill
  \begin{subfigure}[t]{0.32\textwidth}
  \centering
        \captionsetup{skip=-3pt}  
    \includegraphics[width=\linewidth]{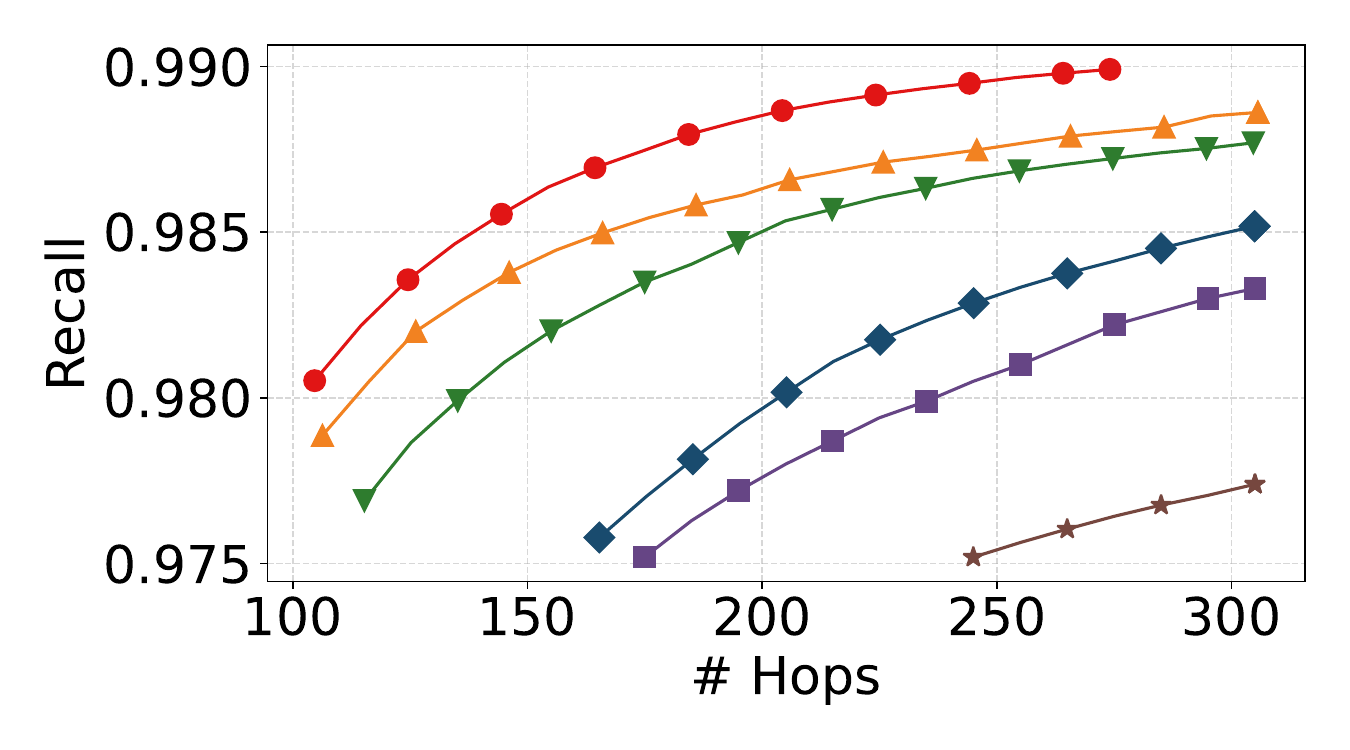}
    \caption{\ttt{LAION-I2I}}
  \end{subfigure}
  \hfill
    \begin{subfigure}[t]{0.32\textwidth}
    \centering
        \captionsetup{skip=-3pt}  
    \includegraphics[width=\linewidth]{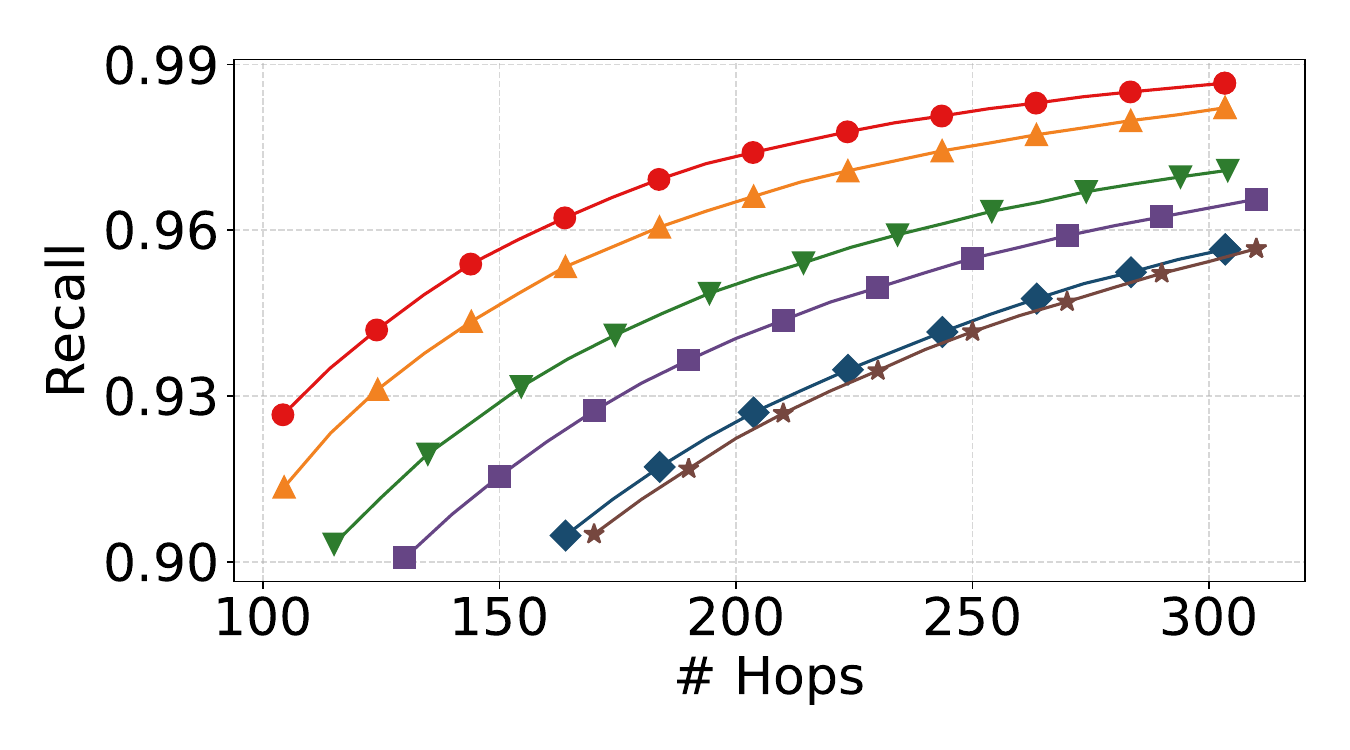}
    \caption{\ttt{GIST}}
  \end{subfigure}
  \caption{Recall@100 vs. number of hops}
  \label{fig:hops-vs-recall}
\end{figure*}

\begin{table*}[ht]
\footnotesize
\centering
\figcapup\figcapup
\caption{Speedups of our $\alpha$-CNG in NDC and \# hops over the best-performing baseline (in bold). For each dataset, all methods reached the same recall@100, i.e., 0.99 for \ttt{SIFT} and \ttt{LAION-I2I}, 0.95 for \ttt{GIST} at 0.95, and 0.90 for \ttt{WIKI}, \ttt{CRAWL}, and \ttt{MSONG}.}
\figcapdown\figcapdown
\label{tab:query-improvements}
\resizebox{\textwidth}{!}{
\begin{tabular}{lcccccc|cccccc}
\toprule
\textbf{}
& \multicolumn{6}{c}{\textbf{NDC}} 
& \multicolumn{6}{c}{\textbf{\# Hops}} \\
\cmidrule(lr){2-7} \cmidrule(lr){8-13}
\textbf{Method}
& \ttt{SIFT} & \ttt{CRAWL} & \ttt{WIKI} & \ttt{MSONG} & \ttt{LAION-I2I} & \ttt{GIST}
& \ttt{SIFT} & \ttt{CRAWL} & \ttt{WIKI} & \ttt{MSONG} & \ttt{LAION-I2I} & \ttt{GIST} \\
\midrule
\ttt{HNSW}
& 4224 & 10284 & 5204 & \bf{12365} & N/A & 8179 
& 171 & 190 & 244 & 315 & N/A & 281
\\

\ttt{Vamana}
& 4373 & \bf{9286} & 4966 & 13150 & 11523 & 8177
& \bf{166} & \bf{172} & 223 & \bf{281} & \bf{810} & 232\\

\ttt{NSG}
& 3903 & 14837 & 4731 & 18289 & \bf{10297} & 7677 
& 175 & 233 & 210 & 456 & 904 & 272
\\

\ttt{$\tau$-MNG}
& \bf{3842} & 10599 & \bf{4556} & 14402 & 13776 & \bf{6971} 
& 177 & 239 & \bf{189} & 319 & 1004 & \bf{199}
\\


\textbf{\ttt{$\alpha$-CNG}} 
& 3911 & 4067 & 4166 & 10438 & 6829 & 6024 
& 143 & 102 & 124 & 167 & 282 & 137
\\
\midrule
\textbf{Speedup} 
& 0.98x & 2.28x & 1.09x & 1.18x & 1.51x & 1.16x    
& 1.16x & 1.69x & 1.52x & 1.68x & 2.88x & 1.45x
\\
\bottomrule
\end{tabular}
}

\end{table*}

\noindent{\bf Recall vs. NDC.} We first assessed the trade-off between (average) recall and (average) NDC on the six datasets at the 1M scale. Figure~\ref{fig:NDC-vs-recall} reports the results obtained by varying the queue size $L$ of beam search. Our \texttt{$\alpha$-CNG} outperformed all baselines, achieving higher recall with fewer NDCs. The only exception is the \ttt{SIFT} dataset, for which all methods exhibited strong performance due to its low intrinsic dimensionality \cite{pcc+23}. \texttt{Fixed-$\alpha$-CNG} also surpassed the four baselines, except on \ttt{MSONG}, but was outperformed by \texttt{$\alpha$-CNG}. 

To provide a more detailed comparison, we report the NDC and number of hops at fixed recall@100 levels in Table~\ref{tab:query-improvements}. We configure different recall levels for the datasets based on their query performance in Figure~\ref{fig:NDC-vs-recall}, i.e., 0.99 for easy datasets \texttt{SIFT} and \texttt{LAION-I2I}, 0.95 for \texttt{GIST}, and 0.90 for difficult datasets \texttt{CRAWL}, \texttt{WIKI}, and \texttt{MSONG}. As shown in Table~\ref{tab:query-improvements}, our $\alpha$-CNG reduces NDC by over 15\% compared to the best-performing baseline on four out of the six datasets, and the maximum speedup can be 2.28x. We also observe that none of the baselines consistently outperformed the others, making our $\alpha$-CNG a good choice due to its consistently good performance across different datasets.
    
\extraspacing{\bf Recall vs. \# hops.} We next evaluated the (average) number of hops by varying $L$. Figure~\ref{fig:hops-vs-recall} shows that \texttt{$\alpha$-CNG} substantially reduced the number of hops across all datasets, and \texttt{Fixed-$\alpha$-CNG} also outperformed the baselines on most datasets. \rev{Table~\ref{tab:query-improvements} confirms that the reductions in the number of hops over the best-performing baseline are substantial: exceeding 45\% on five out of the six datasets, with a maximum speedup of 2.88x.} These gains arise from two key factors: $\alpha$-CNG approximates $\alpha$-CG for faster convergence, and our adaptive pruning strategy preserves more long-distance shortcut edges, effectively reducing the number of hops. The reduced hops of $\alpha$-CNG are advantageous for disk-based~\cite{sds+19} or distributed deployments of PGs, where hops correspond to I/O operations.

\extraspacing{\bf Scalability.} Finally, we tested the scalability of all competing methods on the 100M-scale datasets (see Figure~\ref{fig:large-data}). Although the improvements in NDC are less pronounced than those observed at 1M-scale, \texttt{$\alpha$-CNG} consistently outperformed all baselines, and \texttt{Fixed-$\alpha$-CNG} matched the best baseline. Both our variants significantly reduced the number of hops (over 40\% when recall@100=0.98), demonstrating that our methods converge quickly on large-scale datasets.

\begin{figure*}[t]
  \centering
  \includegraphics[width=0.8\textwidth]{Figure/legend_public_cng.pdf}

  \captionsetup{skip=1pt}

  \begin{subfigure}[b]{0.8\textwidth}
    \centering
    \includegraphics[width=0.42\linewidth]{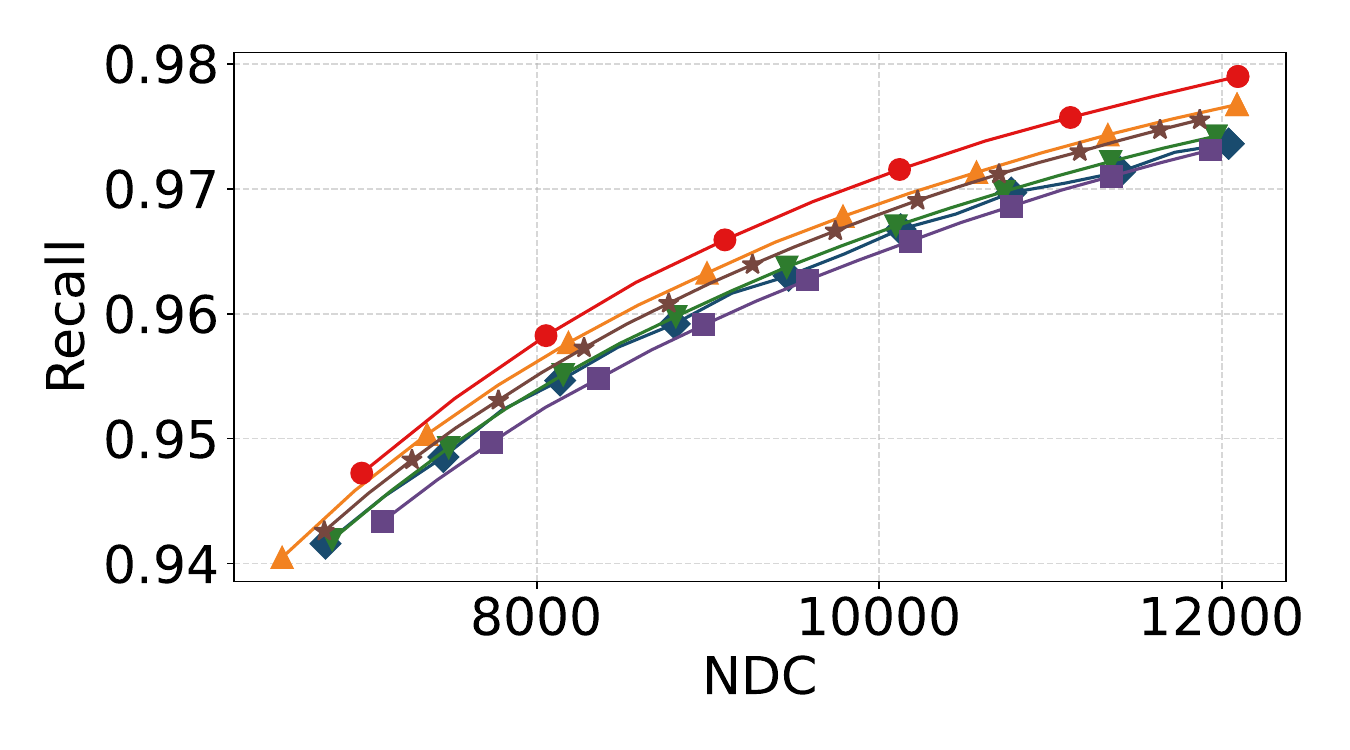}
    \hfill
    \includegraphics[width=0.42\linewidth]{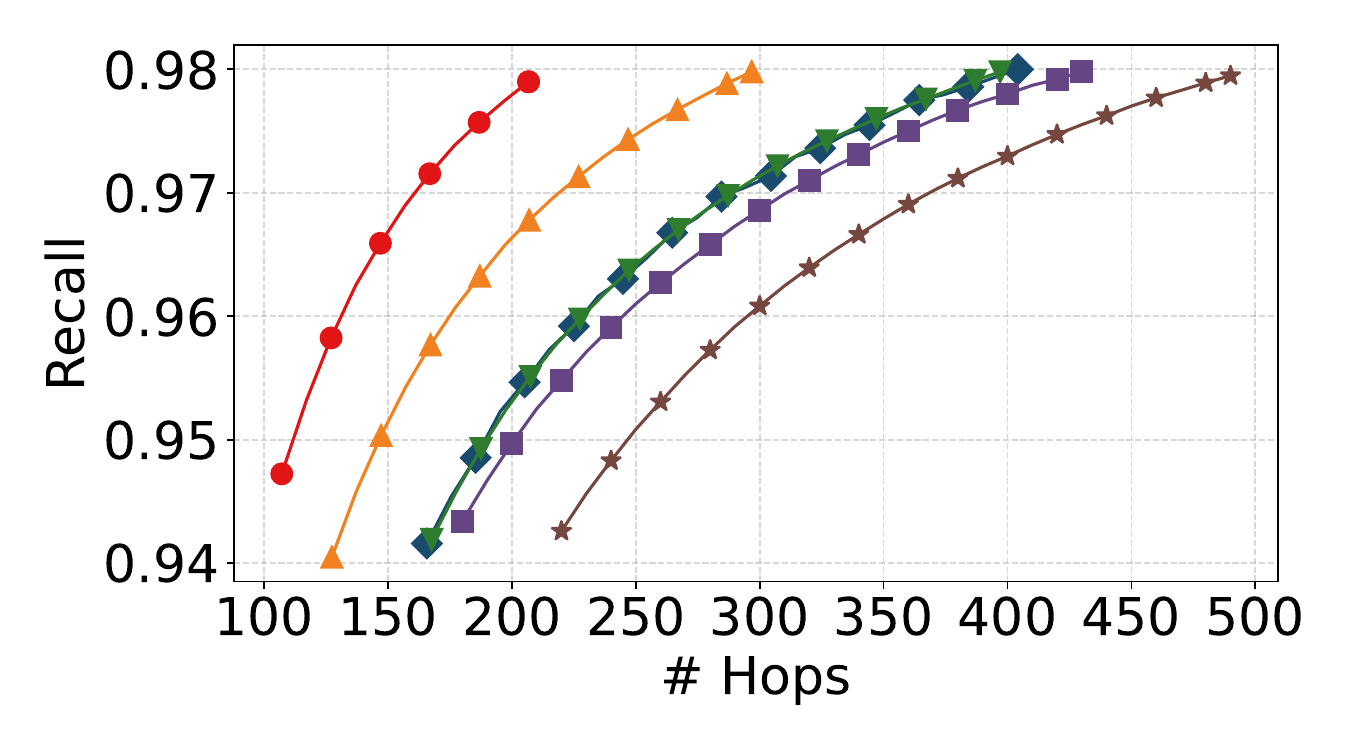}
    \caption{\ttt{DEEP100M}}
  \end{subfigure}

  \begin{subfigure}[b]{0.8\textwidth}
    \centering
    \includegraphics[width=0.42\linewidth]{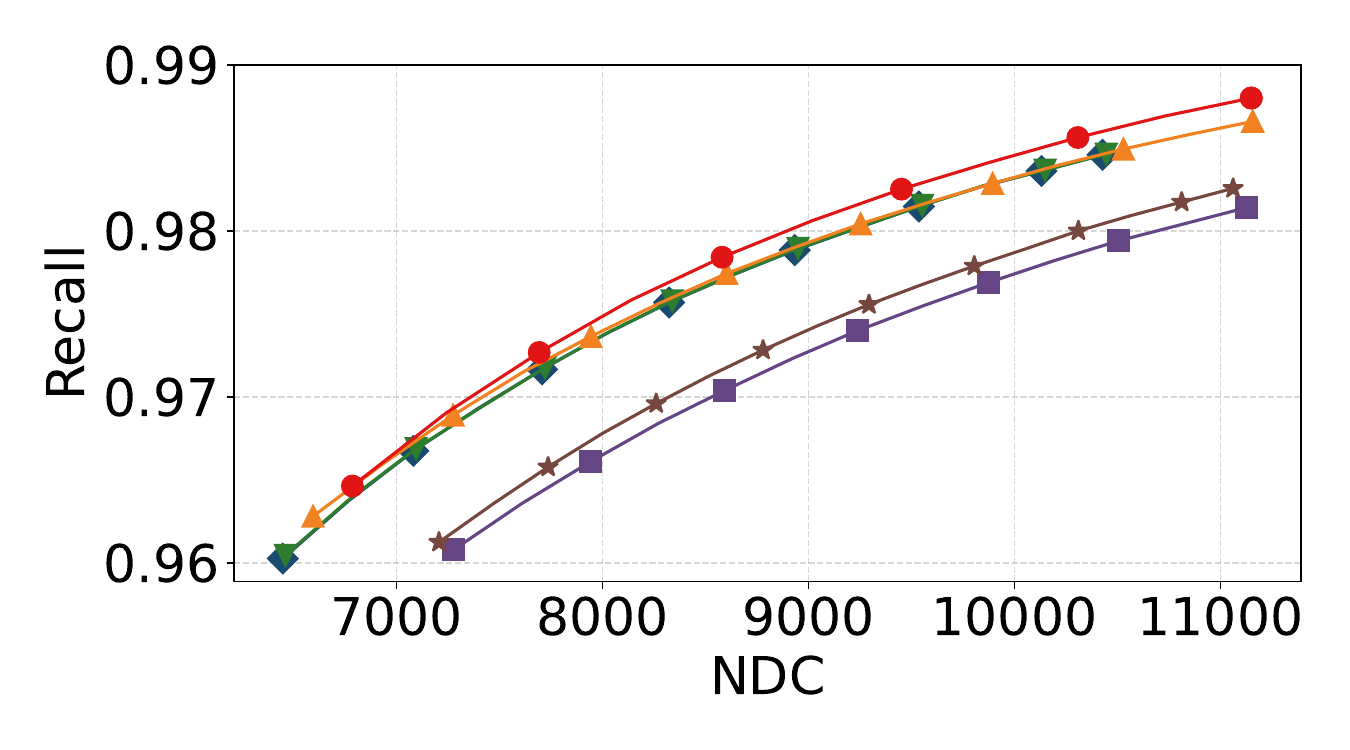}
    \hfill
    \includegraphics[width=0.42\linewidth]{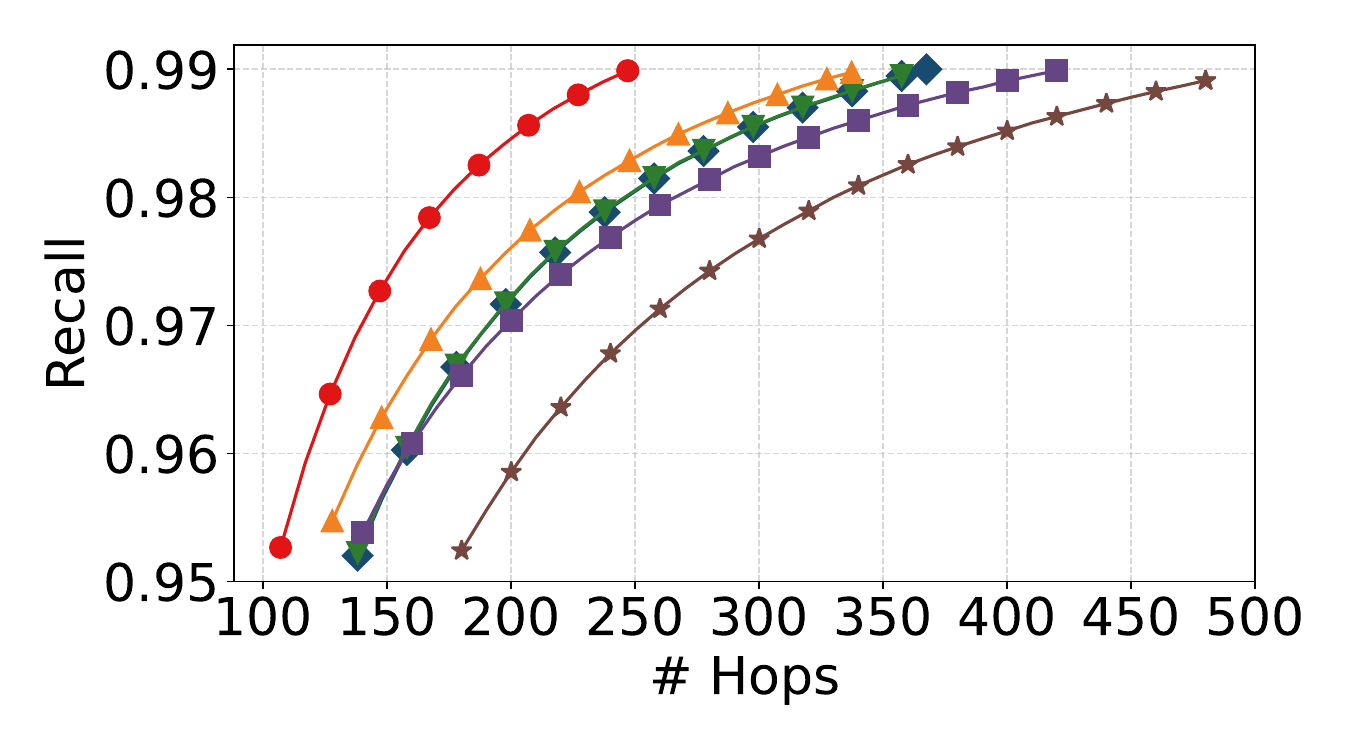}
    \caption{\ttt{BIGANN100M}}
  \end{subfigure}

  \caption{Recall@100 vs. NDC and Recall@100 vs. \# hops on 100M-scale datasets}
  \label{fig:large-data}
\end{figure*}
\subsection{Effect of Parameters} \label{sec:exp:para-scala}

\begin{figure*}[t]  
  \centering
              \captionsetup{skip=1pt}  
  \begin{subfigure}[t]{0.32\textwidth}
  \centering
        \captionsetup{skip=-3pt}  
    \includegraphics[width=1\linewidth]{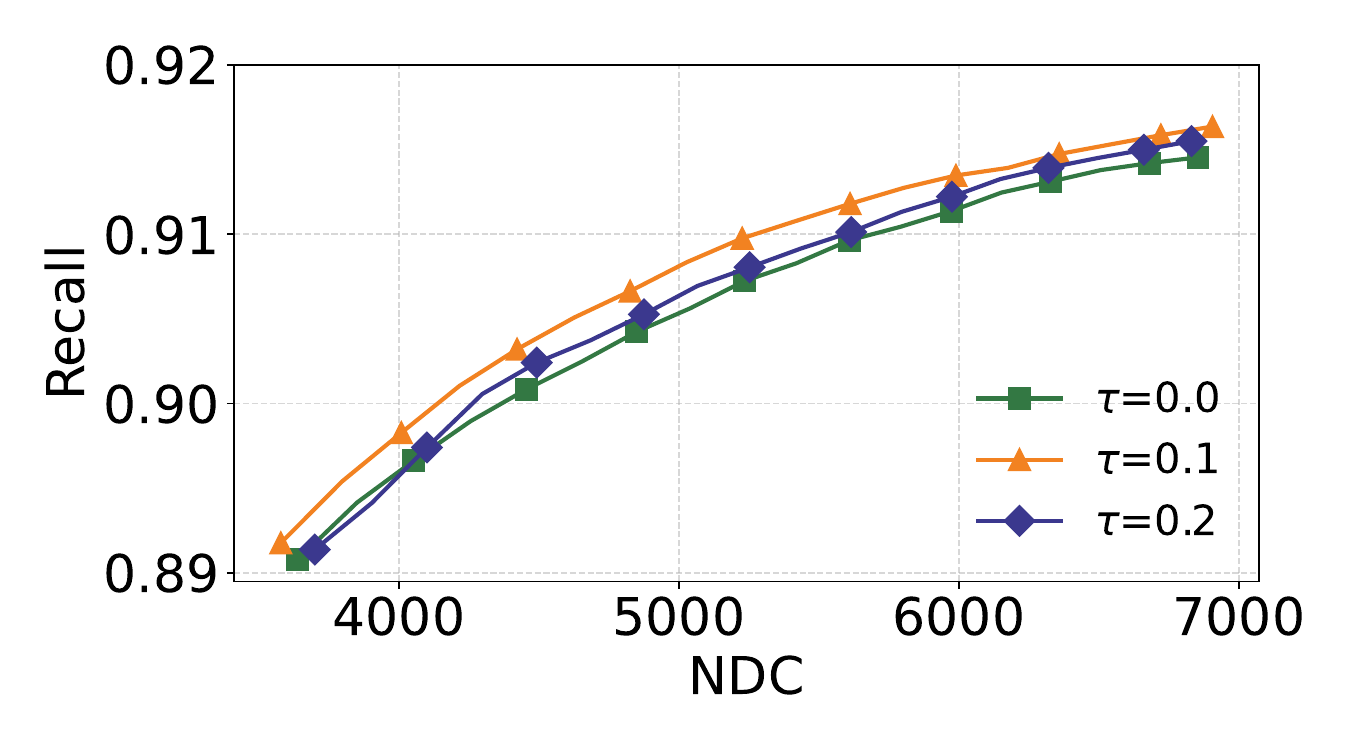}
    \caption{\ttt{WIKI}}
  \end{subfigure}
  \hfill
  \begin{subfigure}[t]{0.32\textwidth}
    \centering
        \captionsetup{skip=-3pt}  
    \includegraphics[width=1\linewidth]{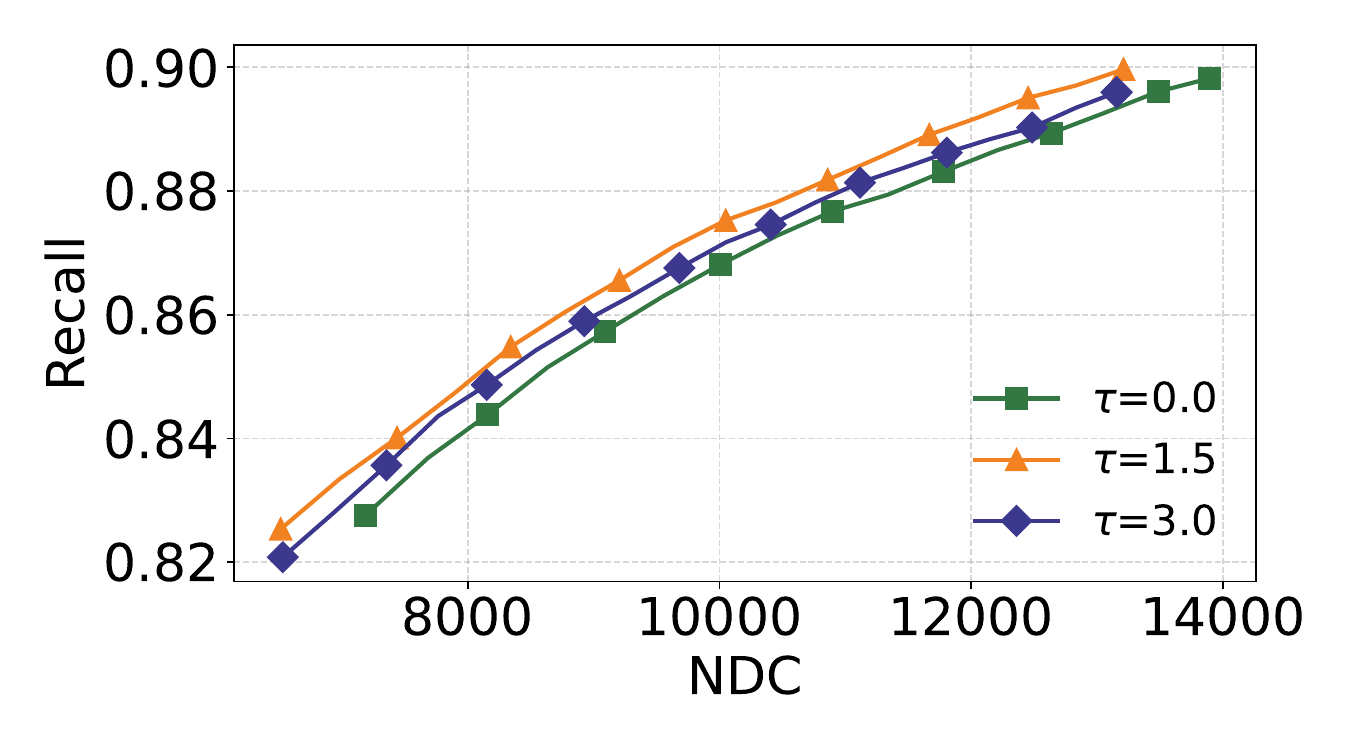}
    \caption{\ttt{MSONG}}
  \end{subfigure}
  \hfill
    \begin{subfigure}[t]{0.32\textwidth}
      \centering
        \captionsetup{skip=-3pt}  
    \includegraphics[width=1\linewidth]{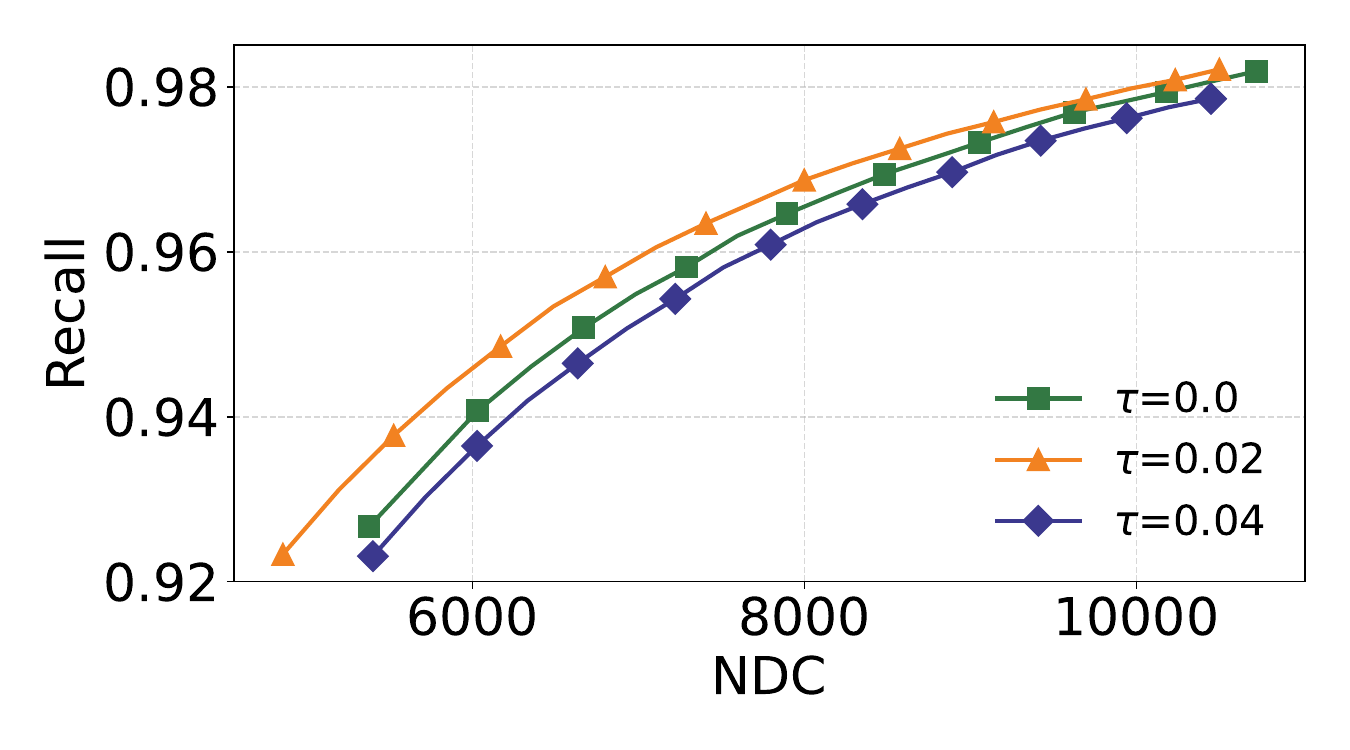}
    \caption{\ttt{GIST}}
  \end{subfigure}
  \caption{Recall@100 vs. NDC with varying $\tau$}
  \label{fig:varying tau}
  \centering
              \captionsetup{skip=1pt}  
    \begin{subfigure}[t]{0.32\textwidth}
      \centering
        \captionsetup{skip=-3pt}  
    \includegraphics[width=1\linewidth]{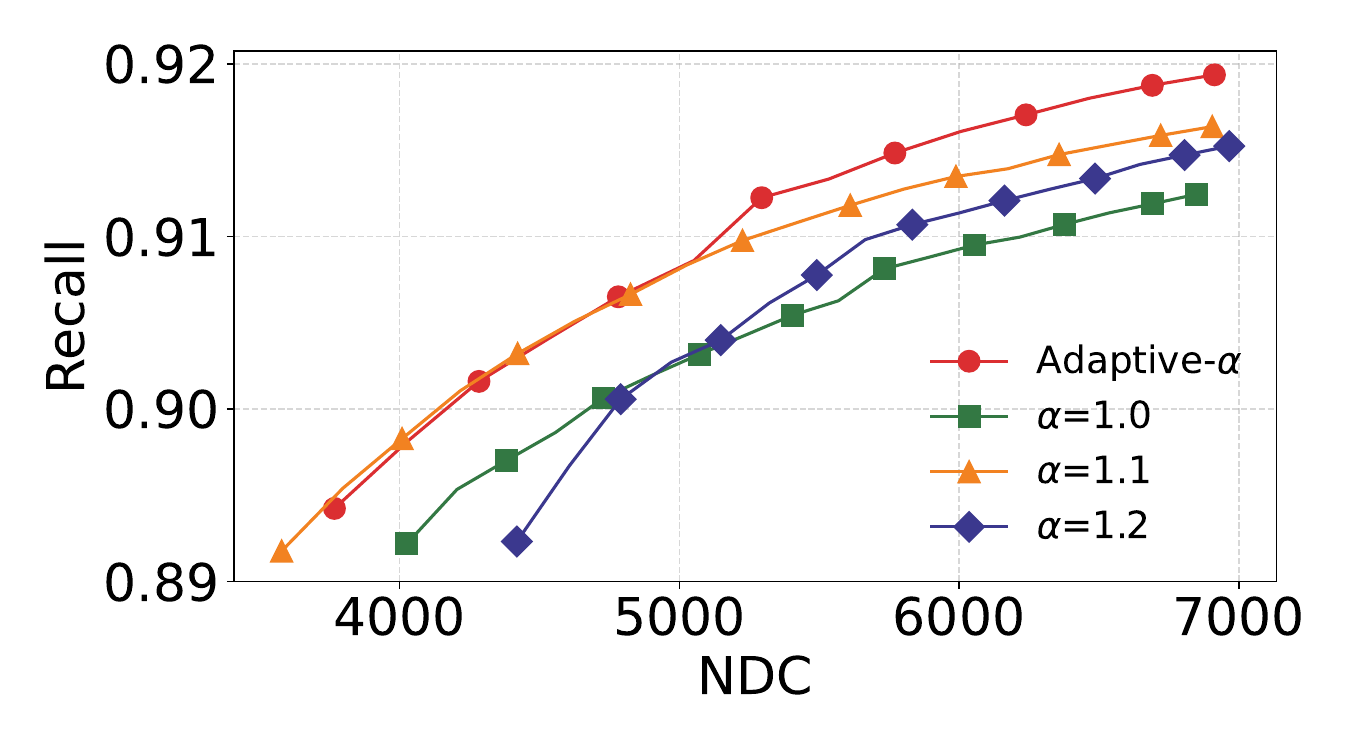}
    \caption{\ttt{WIKI}}
  \end{subfigure}
  \hfill
    \begin{subfigure}[t]{0.32\textwidth}
      \centering
        \captionsetup{skip=-3pt}  
    \includegraphics[width=1\linewidth]{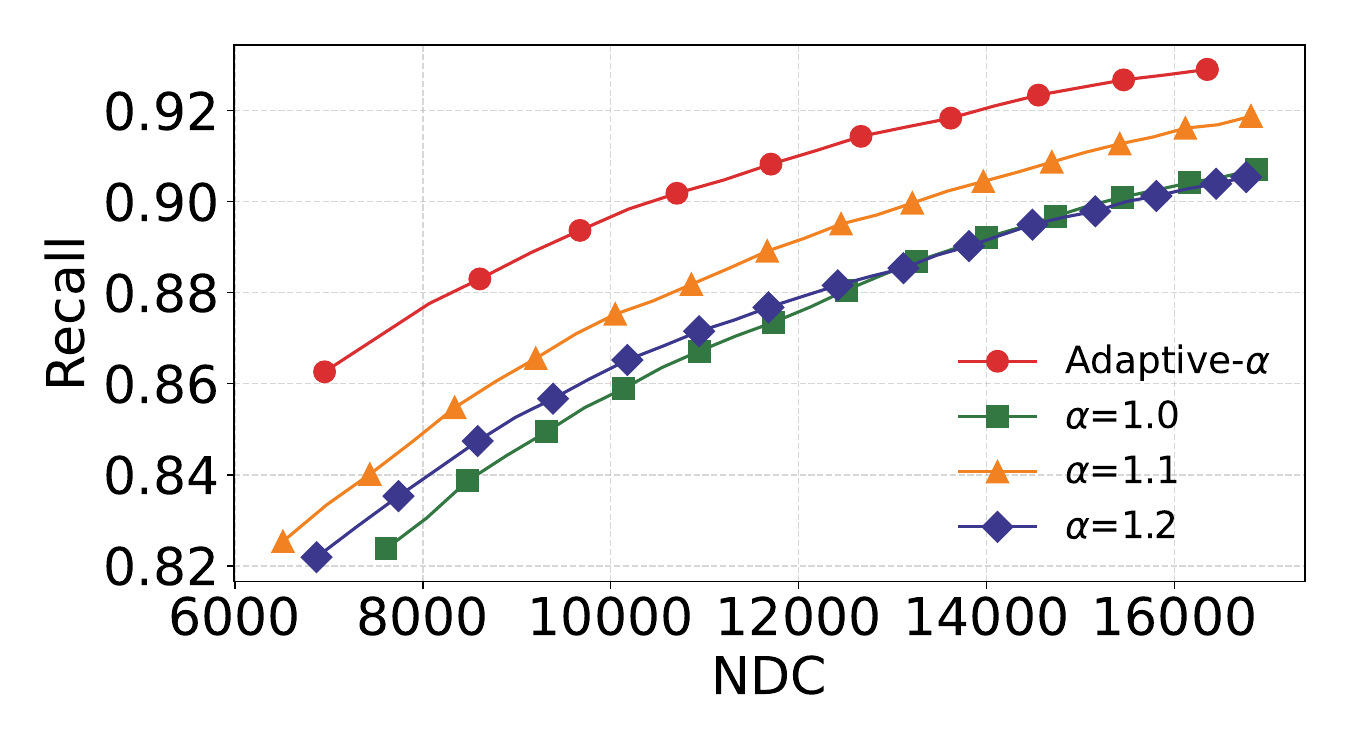}
    \caption{\ttt{MSONG}}
  \end{subfigure}
  \hfill
  \begin{subfigure}[t]{0.32\textwidth}
    \centering
        \captionsetup{skip=-3pt}  
    \includegraphics[width=1\linewidth]{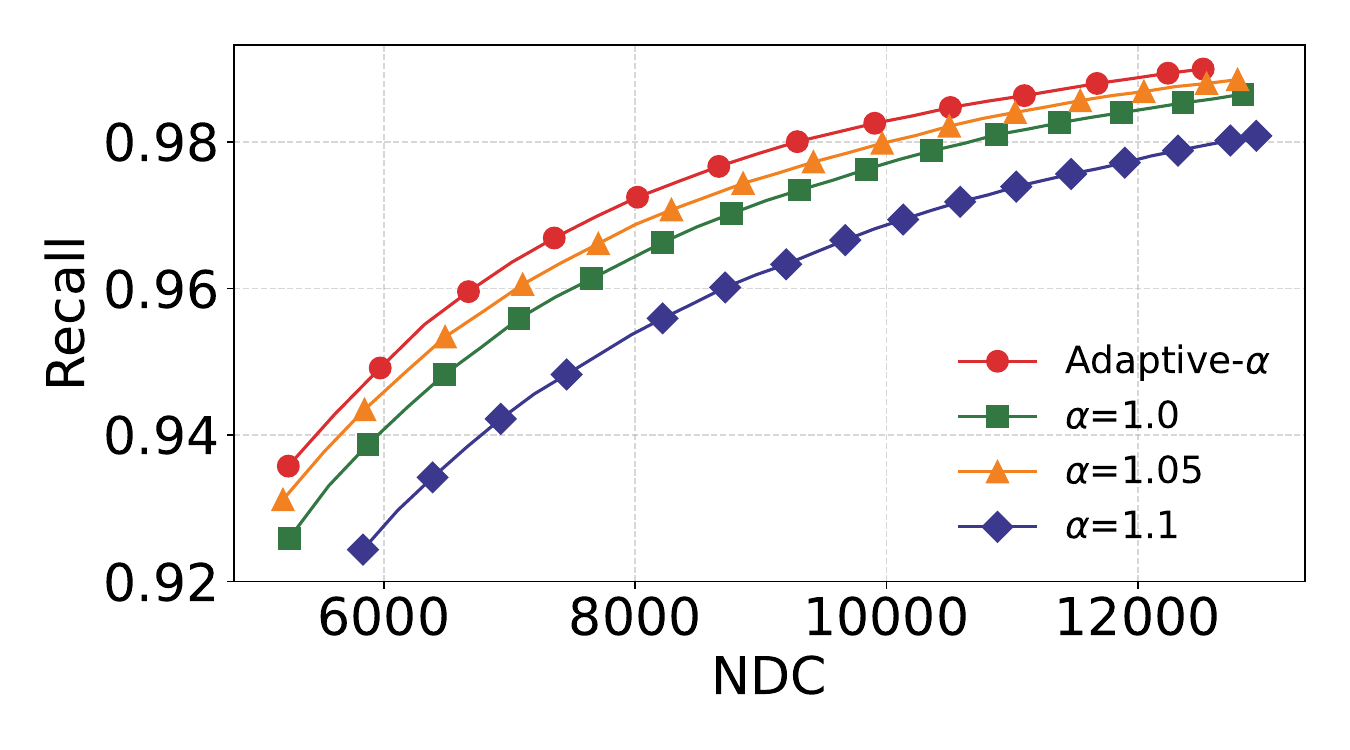}
    \caption{\ttt{GIST}}
  \end{subfigure}
  \caption{Recall@100 vs. NDC with varying $\alpha$}
  \label{fig:varying alpha}
\end{figure*}

This subsection explores the effects of $\tau$ and $\alpha$ in our edge pruning rule (Definition~\ref{def:edge-pruning}). We evaluated \texttt{Fixed-$\alpha$-CNG} on three 1M-scale datasets with the highest dimensionality from different data types: \ttt{WIKI} for text, \ttt{MSONG} for audio, and \ttt{GIST} for images.

\extraspacing{\bf Varying $\tau$.} Figure~\ref{fig:varying tau} illustrates the impact of the parameter $\tau$ on search performance, where the middle tested $\tau$ for each dataset is the optimal $\tau$ identified. We observe a trend similar to that reported for $\tau$-MNG \cite{pcc+23}: as $\tau$ increases from zero, search performance initially improves before deteriorating with further increases in $\tau$. Recall that the search time is proportional to the total out-degrees of the hop vertices. A moderate increase in $\tau$ allows connecting more candidates and enhances the graph connectivity, thus reducing the number of hops. However, if $\tau$ increases excessively, the out-degree may become too large, as the search must compute distances for more neighbors at each step. Additionally, the increased out-degrees may exceed the limit $M$ for some nodes, resulting in the replacement of long-distance shortcut edges with less useful ones. Overall, a small positive $\tau$ enhances graph connectivity and improves search performance, while excessively large values can lead to over-inclusion and reduced pruning effectiveness.

\extraspacing{\textbf{Varying $\alpha$.}} We evaluated \texttt{Fixed-$\alpha$-CNG} for three values of $\alpha$, along with \texttt{$\alpha$-CNG} (see Figure~\ref{fig:varying alpha}). The middle $\alpha$ value tested is the optimal value chosen for \texttt{Fixed-$\alpha$-CNG}. We observe that increasing $\alpha$ initially enhances the performance of \texttt{Fixed-$\alpha$-CNG}, but it deteriorates as $\alpha$ continues to rise, for the same reasons observed with increasing $\tau$. Notably, \texttt{$\alpha$-CNG} outperformed \texttt{Fixed-$\alpha$-CNG} for all tested $\alpha$ values. This advantage stems from its capability to locally adjust $\alpha$, allowing each node to retain more shortcut edges while ensuring that its degree does not exceed $M$. Overall, \texttt{$\alpha$-CNG} proves to be more practical, as it avoids the need for manual tuning of the parameter $\alpha$ while delivering improved query performance.

\subsection{Construction Performance} \label{sec:exp:construction}

\begin{table*}[ht]
\footnotesize
\centering
\figcapup\figcapup
\caption{Index construction time and index size across six datasets.}
\figcapdown\figcapdown
\label{tab:index_time_size}

\resizebox{\textwidth}{!}{
\begin{tabular}{lcccccc|cccccc}
\toprule
 & \multicolumn{6}{c}{\textbf{Index time (secs)}} 
 & \multicolumn{6}{c}{\textbf{Index size (MB)}} \\
\cmidrule(lr){2-7} \cmidrule(lr){8-13}

\textbf{Method}
& \ttt{SIFT} & \ttt{CRAWL} & \ttt{WIKI} & \ttt{MSONG} & \ttt{LAION-I2I} & \ttt{GIST}
& \ttt{SIFT} & \ttt{CRAWL} & \ttt{WIKI} & \ttt{MSONG} & \ttt{LAION-I2I} & \ttt{GIST} \\
\midrule

\ttt{HNSW}
& 42 & 159 & 91 & 145 & 116 & 206
& 256 & 509 & 256 & 254 & 256 & 256 \\

\ttt{Vamana}
& 9 & 43 & 19 & 38 & 35 & 67
& 163 & 409 & 114 & 151 & 134 & 118 \\

\ttt{NSG}
& 43 & 174 & 97 & 223 & 93 & 200
& 115 & 183 & 113 & 97 & 83 & 90 \\

\ttt{$\tau$-MNG}
& 43 & 182 & 98 & 242 & 98 & 217
& 113 & 329 & 124 & 117 & 127 & 131 \\

\textbf{\ttt{Fixed-$\alpha$-CNG}}
& 42 & 206 & 88 & 215 & 99 & 207
& 158 & 512 & 179 & 203 & 186 & 162 \\

\textbf{\ttt{$\alpha$-CNG}}
& 53 & 221 & 119 & 227 & 110 & 228
& 188 & 563 & 228 & 225 & 246 & 227 \\

\midrule
\textbf{Vector Data Size}
& \multicolumn{6}{c|}{--}
& 492 & 2284 & 1464 & 1597 & 2908 & 3666 \\
\bottomrule
\end{tabular}
}
\end{table*}
\noindent{\bf Index time.} Table~\ref{tab:index_time_size} compares indexing times. Both our methods, \texttt{$\alpha$-CNG} and \texttt{Fixed-$\alpha$-CNG}, achieved performance comparable to that of baseline methods, except for \texttt{Vamana}. Compared to \texttt{Vamana}, our methods employ approximate $K$-NN graphs for candidate generation and a relaxed edge pruning rule that preserves more shortcut edges. This enhances graph connectivity, improving query performance as shown in Figure~\ref{fig:NDC-vs-recall}.

\vgap

We further observe the following regarding our two methods: (1) Although using a more relaxed pruning, \texttt{Fixed-$\alpha$-CNG} exhibited index times comparable to \texttt{NSG} and \texttt{$\tau$-MNG}, and occasionally outperformed them. This efficiency stems from our lazy pruning strategy during backward edge insertion, reducing invocations of the pruning procedure. (2) Despite \texttt{$\alpha$-CNG} iteratively pruning candidate sets for multiple $\alpha$ values, its index time remains only marginally higher than \texttt{Fixed-$\alpha$-CNG}. This is due to our distance-reusing mechanism that optimizes the total number of distance computations.

\extraspacing{\bf Index Size.} Table~\ref{tab:index_time_size} presents the sizes of the indexes and the corresponding raw vector data. The index sizes of our methods are comparable to those of \texttt{HNSW} but larger than the other baselines. This is because we retain additional long-distance shortcut edges to accelerate the convergence to ANN and enhance query performance, as validated in Section~\ref{sec:exp:search}. Since the PGs are usually much smaller than the raw vectors, our methods yield only a marginal increase in the overall data size.

\subsection{Effectiveness of Our Edge Pruning Rule} \label{sec:exp:our-rule}

Our two methods, \texttt{NSG}, and \texttt{$\tau$-MNG}, share a framework that generates candidate sets from approximate $K$-NN graphs. Section~\ref{sec:exp:search} shows the efficacy of our pruning rule (Definition~\ref{def:edge-pruning}) within this framework. We further evaluated its effectiveness on two widely used PGs --- \texttt{HNSW} and \texttt{Vamana} --- by creating two variants, \texttt{HNSW+} and \texttt{Vamana+}, through the integration of our edge pruning rule. Figure~\ref{fig:integration into existing systems} presents the recall-NDC trade-offs, with \texttt{$\alpha$-CNG} included for reference. Both \texttt{HNSW+} and \texttt{Vamana+} consistently outperformed their original counterparts. \texttt{HNSW+} achieved more substantial gains over \texttt{HNSW} and even surpassed \texttt{$\alpha$-CNG} on \texttt{MSONG}. We attribute this difference to the distinct pruning strategies used by \texttt{Vamana} and \texttt{HNSW}: \REV{\texttt{Vamana} employs a more relaxed
pruning rule by setting
$r = \dis(p,u)/\alpha$ (Fig.~\ref{fig:lune-final}b), which incorporates the parameter $\alpha$, while \texttt{HNSW} uses the strict pruning rule with
$r = \dis(p,u)$ (Fig.~\ref{fig:lune-final}a). These experimental results demonstrate the transferability of our edge pruning rule to enhance the query performance of other PG solutions.}

Although \texttt{$\alpha$-CNG}, \texttt{HNSW+}, and \texttt{Vamana+} 
share the same $\alpha$-pruning rule, \texttt{$\alpha$-CNG} still achieves 
better performance because it constructs candidate neighbor sets of higher quality. 
Specifically, \texttt{$\alpha$-CNG} first builds a $K$-NN graph and performs 
edge pruning over approximate nearest neighbors obtained from that graph. 
In contrast, \texttt{HNSW+} and \texttt{Vamana+} incrementally insert nodes, 
so each node can only access neighbors that have already been inserted, 
leading to a limited and often suboptimal candidate set for pruning. 
Consequently, the edges retained by \texttt{$\alpha$-CNG} better preserve global connectivity, resulting in higher recall and more 
efficient search.
\begin{figure*}[t]  
  \centering
  \captionsetup{skip=1pt}  
    \includegraphics[width=0.65\textwidth]{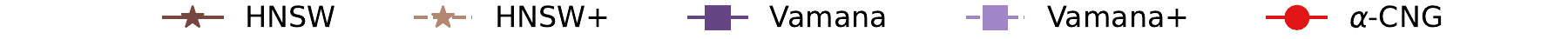}
    
      \begin{subfigure}[t]{0.32\textwidth}
                \captionsetup{skip=-3pt}  
    \includegraphics[width=\linewidth]{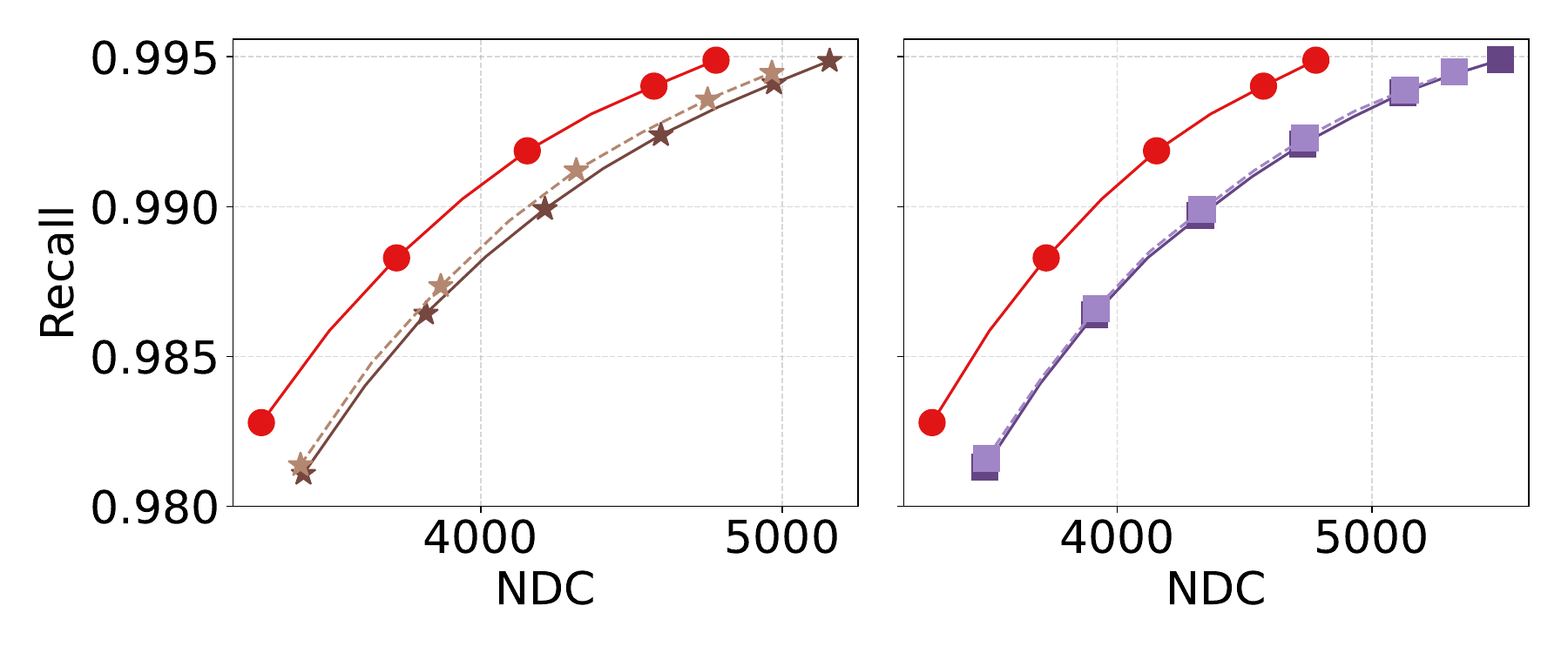}
    \caption{\ttt{SIFT}}
  \end{subfigure}
  \hfill
    \begin{subfigure}[t]{0.32\textwidth}
              \captionsetup{skip=-3pt}  
    \includegraphics[width=\linewidth]{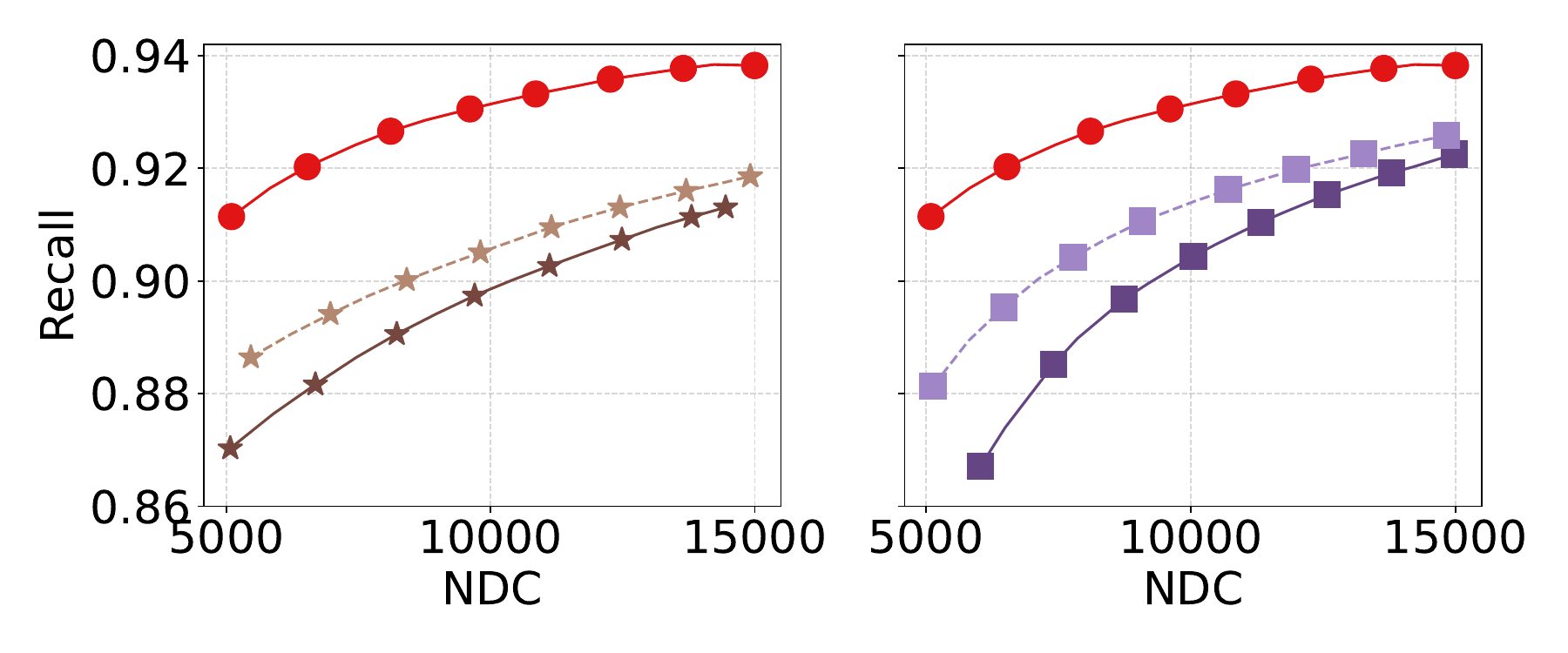}
    \caption{\ttt{CRAWL}}
  \end{subfigure}
  \hfill
  \begin{subfigure}[t]{0.32\textwidth}
  \captionsetup{skip=-3pt}  
    \includegraphics[width=\linewidth]{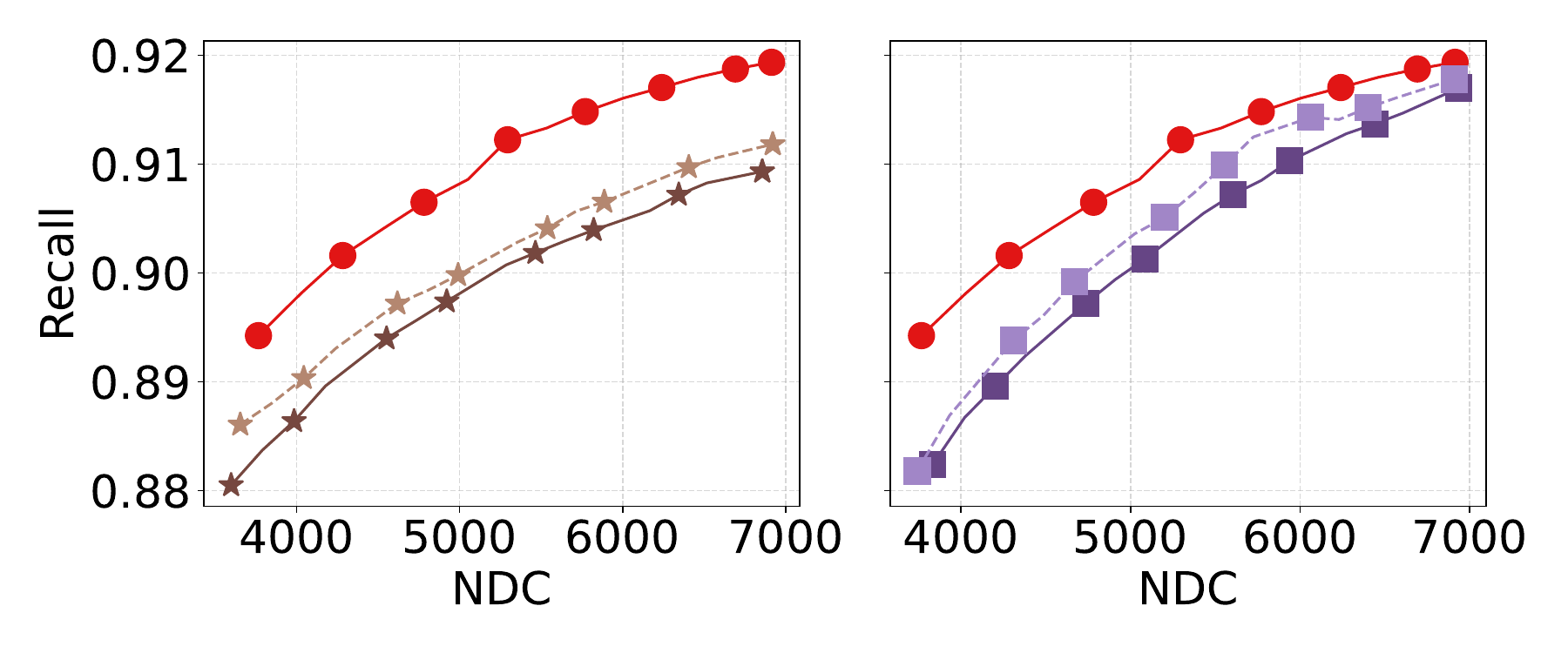}
    \caption{\ttt{WIKI}}
  \end{subfigure}
    \hfill
    \begin{subfigure}[t]{0.32\textwidth}
    \captionsetup{skip=-3pt}  
    \includegraphics[width=\linewidth]{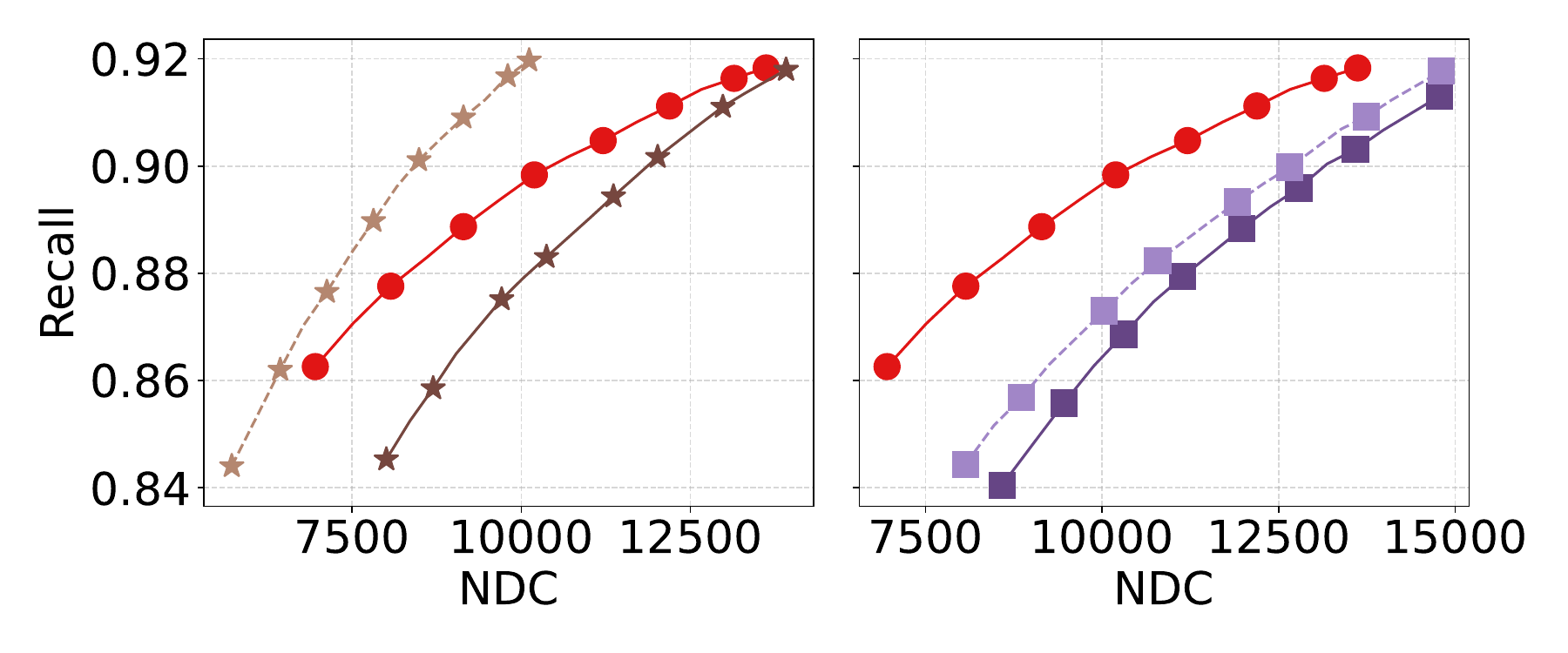}
    \caption{\ttt{MSONG}}
  \end{subfigure}
  \hfill
    \begin{subfigure}[t]{0.32\textwidth}
    \captionsetup{skip=-3pt}  
    \includegraphics[width=\linewidth]{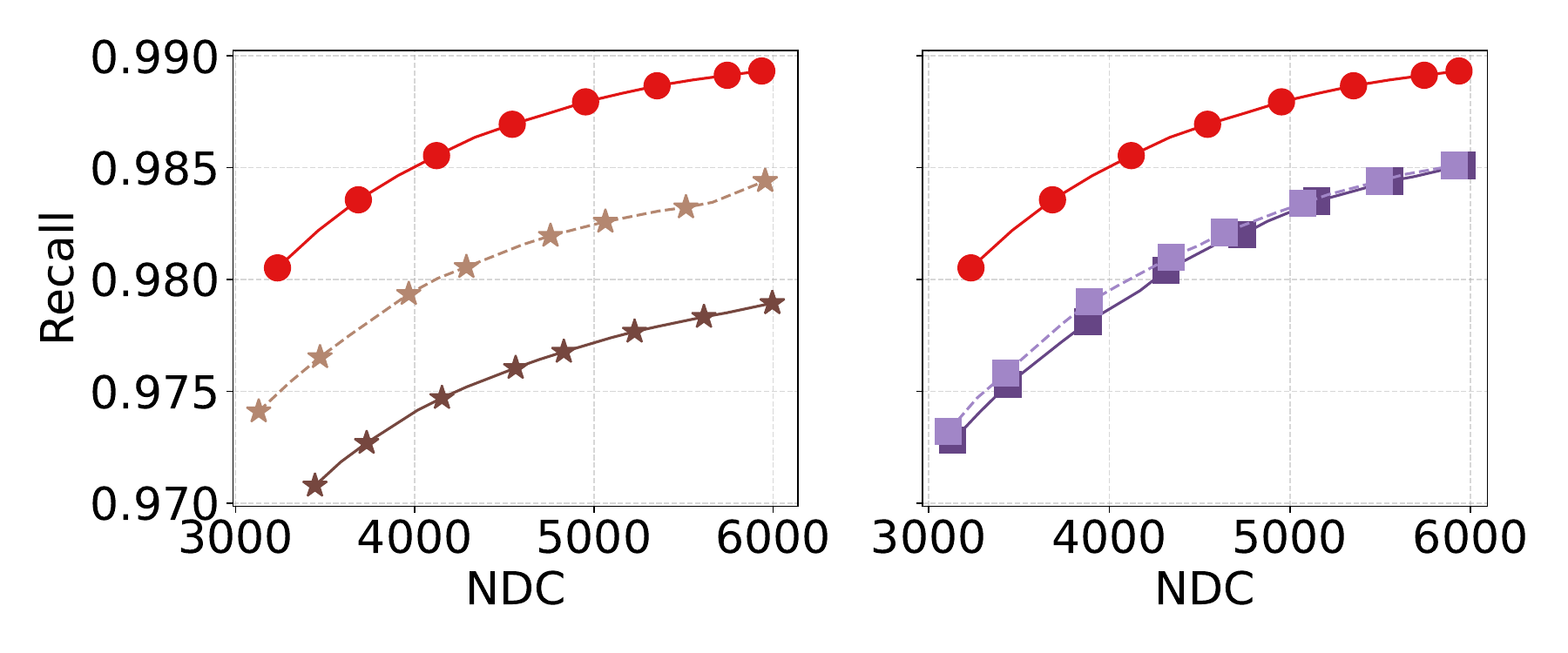}
    \caption{\ttt{LAION-I2I}}
  \end{subfigure}
  \hfill
  \begin{subfigure}[t]{0.32\textwidth}
  \captionsetup{skip=-3pt}  
    \includegraphics[width=\linewidth]{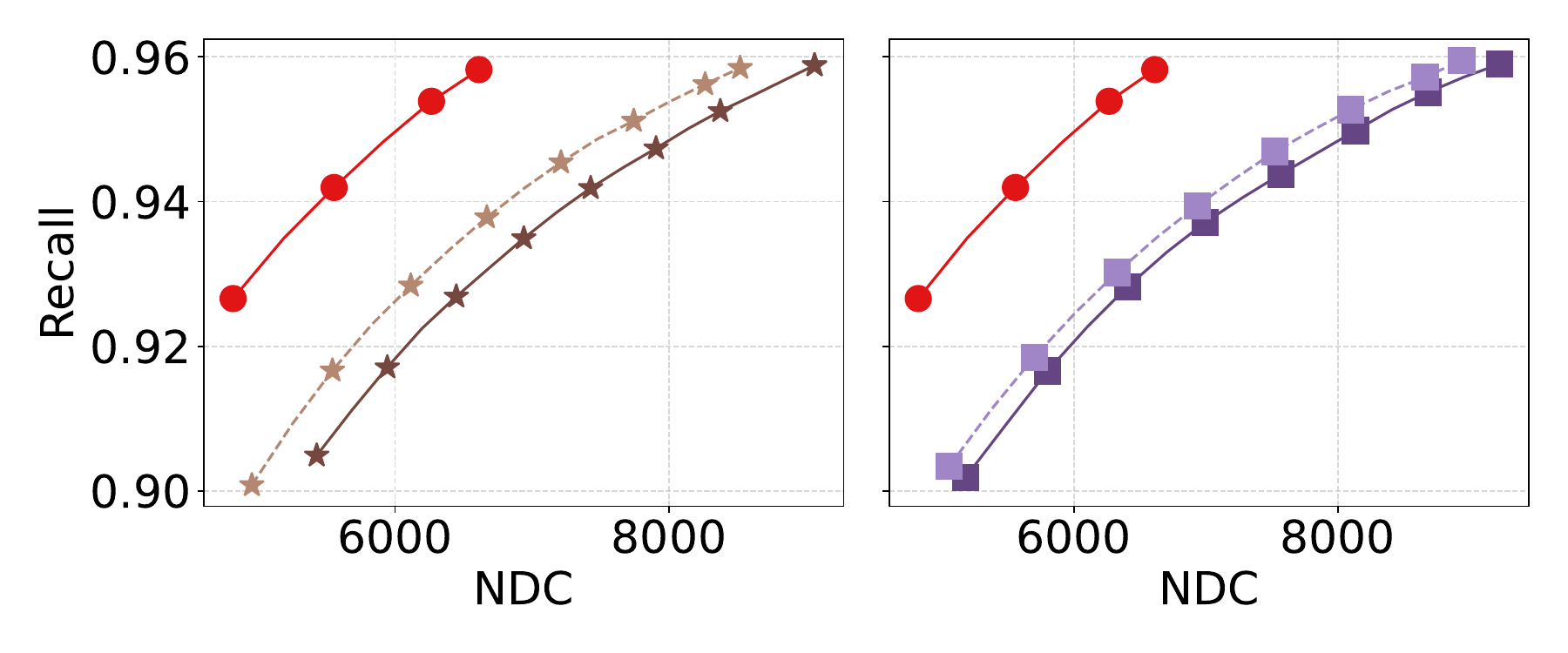}
    \caption{\ttt{GIST}}
  \end{subfigure}
  \caption{Effect of integrating our edge pruning rule into HNSW and Vamana}  
  \label{fig:integration into existing systems}
\end{figure*}

\section{Related Work}

Approximate nearest neighbor (ANN) search has attracted considerable attention over the past two decades, leading to the development of diverse methodologies aimed at enhancing search performance. Recent experimental studies \cite{abf20, lzah20, lzs+19, wxyw21, aep25} indicate that PG-based approaches surpass other techniques, including hash-based methods \cite{im98,swq+14,hfz+15,wplp24,tzz24,lk20,lwwk20}, tree-based methods \cite{ks97,b75,c08,y93,wwj+13,laja08}, and inverted index-based methods \cite{bl14,jds10}.

\extraspacing{\bf PG-based methods.} The Delaunay Graph (DG) \cite{a91}, one of the earliest PGs, is the dual graph of the Voronoi diagram with $O(n^{\lc m/2 \rc})$ space in the $m$-dimensional space. DG guarantees finding the exact NN but suffers from high out-degrees and unbounded search time.

\vgap

Several PGs in the literature provide non-trivial query accuracy guarantees while maintaining graph sparsity. Inspired by the relative neighborhood graph \cite{jt92}, Arya et al. \cite{am93} and Fu et al. \cite{fxwc19} proposed MRNG. However, it can find the exact NN only when $q\in P$. Fu et al. \cite{fwc22} introduced the satellite system graph (SSG) to support the case when $q\notin P$ and the input is randomly distributed, although the worst-case time remains unbounded. Harwood et al. \cite{hd16} studied the scenario when $\delta(q, v^*)$ is at most a constant $\tau > 0$; when $P$ is uniformly drawn from $\real^m$, their PG can find the exact NN in $O(n^{2/m}(\ln n)^2)$ time. Later, $\tau$-MNG \cite{pcc+23} improved the search time to $O(n^{1/m}(\ln n)^2)$. Recently, Indyk et al. proved that the slow preprocessing version of Vamana guarantees to find an $(\frac{\alpha+1}{\alpha-1}+\epsilon)$-ANN of $q$ in $O(\alpha^d \cdot \log \Delta \cdot \log_\alpha \frac{\Delta}{(\alpha-1)\epsilon})$ time, regardless of whether $\delta(q, v^*) \le \tau$. Refer to other theoretical works studying cases where $q\in P$ \cite{dgm+24} or the data follows specific distributions \cite{l18,ps20}.

\vgap

The worst-case construction time of all the aforementioned PGs is $\Omega(n^2)$. To reduce construction time, numerous practical PGs have been proposed in the literature (see \cite{zxc+23, sds+19, fxwc19, my20, hd16, gl23, lkxi21, pcc+23, zth+23, czk+24, wxy+24, ggxl25, yxl+24} and the references therein). A recent survey by Azizi et al. \cite{aep25} identifies several design paradigms for PGs and indicates that neighbor selection (Definition~\ref{def:candidate-sel}) is crucial for improving search performance, warranting further theoretical exploration. This paper aims to contribute to the understanding of this direction.

\vgap


\extraspacing{\bf Additional directions.} Research has explored methods to extend ANN search by incorporating attribute constraints, enabling data retrieval that satisfies both vector similarity and user-specified attributes \cite{gks+23,wlx+23,pkgz24,zqz+24,cscz24,xgg+24}. Other studies have focused on the efficient construction of PGs \cite{zth+23,dcl11,om23,yxl+24} and on supporting updates \cite{ssks21,xmb+25,www+20}. Recently, quantization-based techniques \cite{jds10,gsl+20,pel+22,abh+23,gl24,ggy+25} have been developed to compress high-dimensional vectors, thereby accelerating distance computations. These techniques can be integrated into PGs to enhance both construction and search efficiency, as investigated in recent research \cite{ggxl25,ssks21}.

\section{Conclusions} \label{sec:conc}
This paper introduces $\alpha$-CG, a new PG structure for high-performance ANN search. Specifically, $\alpha$-CG employs a well-designed pruning rule to eliminate ineffective candidates. We prove that, under a realistic assumption that the distance between the query and its NN is bounded by a constant $\tau$, $\alpha$-CG guarantees exact NN retrieval in poly-logarithmic time. Without this assumption, it ensures ANN search within the same complexity bounds. To reduce graph construction overhead, we develop an approximate variant $\alpha$-CNG with an adaptive local pruning rule that avoids manually tuning the parameter $\alpha$ and preserves more useful shortcut edges. We also propose optimizations to accelerate graph construction further. Empirical results show $\alpha$-CNG consistently outperforms state-of-the-art PG methods, achieving superior accuracy-efficiency trade-offs. Furthermore, our edge pruning rule demonstrates transferability, enhancing query performance when integrated into other popular PGs.

\begin{acks}
	Shangqi Lu's research was supported by a start-up fund from HKUST-Guangzhou.
\end{acks}

\bibliographystyle{ACM-Reference-Format}
\bibliography{sigmod}
\end{document}

%% file: def/yf-def.tex
\usepackage{color}
\usepackage{euscript}
\usepackage{graphicx}
\usepackage{mathrsfs} 
\usepackage{microtype}
\usepackage[normalem]{ulem}
\usepackage{wrapfig}

\def\ttt{\texttt}

\def\extraspacing{\vspace{2mm} \noindent}
\def\figcapup{\vspace{-1mm}}
\def\figcapdown{\vspace{-0mm}}

\def\vgap{\vspace{1mm}}

\newtheorem{theorem}{Theorem}
\newtheorem{lemma}[theorem]{Lemma}

\newtheorem{definition}{Definition}

\newcommand{\myitems}[1]{\begin{itemize}\setlength #1 \end{itemize}}

\newcommand{\myeqn}[1]{\begin{eqnarray}#1\end{eqnarray}}

\newcommand{\set}[1]{\{#1\}}

\newcommand{\rev}[1]{#1}

\def\mit{\mathit}

\def\eps{\epsilon}
\def\fr{\frac}
\def\-{\mbox{-}}

\def\real{\mathbb{R}}

\def\lc{\lceil}

\def\rc{\rceil}

\def\nn{\nonumber}

\def\*{\star}

